\documentclass[journal,onecolumn,draftcls]{IEEEtran}

\usepackage[utf8]{inputenc} 
\usepackage[T1]{fontenc}
\usepackage{url}
\usepackage{ifthen}
\usepackage{cite}
\usepackage[cmex10]{amsmath} 
\usepackage{amssymb,mathrsfs,amsfonts,amsthm}
\usepackage{bbm}
\usepackage{algorithmic}
\usepackage{graphicx}
\usepackage{tabularx}
\usepackage{textcomp}
\usepackage{xcolor}
\usepackage{soul}
\usepackage{cleveref}
\usepackage{mathabx}
\usepackage{comment}
\usepackage{balance}
\usepackage{tikz}
\usepackage{mathtools}

\makeatletter%
\if@twocolumn%

\newcommand{\includeonetwocol}[2]{#2}
\def\twocolbreak{\nonumber\\ &}%
\def\twocolAlignMarker{&}%
\else

\newcommand{\includeonetwocol}[2]{#1}
\def\twocolbreak{}%
\def\twocolAlignMarker{}%
\fi%
\makeatother%

\newtheorem{theorem}{Theorem}
\newtheorem{lemma}{Lemma}

\newtheorem{proposition}{Proposition}

\newtheorem{corollary}{Corollary}

\def\BibTeX{{\rm B\kern-.05em{\sc i\kern-.025em b}\kern-.08em
    T\kern-.1667em\lower.7ex\hbox{E}\kern-.125emX}}
\allowdisplaybreaks

\begin{document}

\title{The Transmitter Trojan Channel: Achievable Rates and Error Exponents}

\author{
Maryam Farahnak-Ghazani and Aria Nosratinia%
\thanks{This work has been submitted to the IEEE for possible publication. Copyright may be transferred without notice, after which this version may no longer be accessible.}
\thanks{
Maryam Farahnak-Ghazani and Aria Nosratinia are with the
department of Electrical and Computer Engineering, University of Texas at Dallas, Richardson, TX, USA
(e-mail: maryam.farahnak@utdallas.edu; aria@utdallas.edu).
}
\thanks{This work was supported in part by the NSF grant 2148211.}
}

\maketitle
\begin{abstract}
Experiments have shown that modifications in communication hardware can enable unauthorized exfiltration of data from transmitters. We study the resulting Transmitter Trojan channel, in which a malicious encoder embedded within a legitimate transmitter perturbs the transmitted codeword on a symbol-by-symbol basis. The model reflects the limited footprint, delay, and computational complexity available to such hardware modifications. We derive achievable 
rogue message rate
and achievable rates for the legitimate message 
under the mismatch induced by the Trojan. We also analyze decoding error exponents for the legitimate and rogue receivers, including operational exponent lower bounds obtained through channel-resolvability transfer arguments. Finally, we study Trojan detection at the legitimate receiver under different levels of statistical knowledge, including likelihood-ratio detection when the ensemble-averaged induced channel is known and goodness-of-fit detection when only the nominal channel is known. The resulting detection error exponents, together with the achievable rate and decoding error exponents, quantify tradeoffs among 
rogue message rate, legitimate message rate, decoding reliability, and Trojan detectability.
We specialize the analysis to the binary symmetric channel and the additive white Gaussian noise channel, and illustrate the resulting tradeoffs numerically.
\end{abstract}


\section{Introduction}

The design and manufacturing process of integrated circuits is increasingly distributed across multiple entities, creating opportunities for malicious modifications in the hardware supply chain~\cite{xiao2016hardware}. In communication transmitters, such modifications can perturb the emitted signal and enable unauthorized data leakage while preserving the apparent operation of the legitimate link~\cite{DirtyConstellation,classen2015practical,liu2016silicon,subraman2019demonstrating,
subramani2020amplitude, dhole2024wireless}. These examples motivate an information-theoretic model for transmitter-side malicious modifications and their impact on leakage, reliability, and detectability.

We model a Transmitter Trojan as a malicious encoder embedded within an otherwise legitimate transmitter. In the uncontaminated system, the transmitter maps a legitimate message to a codeword, and the legitimate receiver is designed for the corresponding channel law. In the Trojan-infested system, the transmitted symbols are modified through a symbol-by-symbol mapping whose arguments come from the legitimate codeword and a Trojan codeword carrying a rogue message. The emitted sequence is then observed by both the legitimate receiver and a rogue receiver. 
The legitimate receiver continues to use the decoder designed for the uncontaminated system. The uncontaminated and Trojan-infested systems are shown in
Fig.~\ref{fig:1_Legitimate} and Fig.~\ref{fig:2_Trojan}, respectively.

\begin{figure}[t!]
    \centering
\includegraphics[scale=0.6]{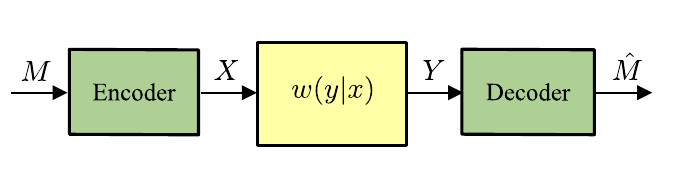}
    \vspace{-10pt}
    \caption{
    Uncontaminated communication system.} 
    \label{fig:1_Legitimate}
\end{figure}
\begin{figure}
    \centering
\includegraphics[scale=0.6]{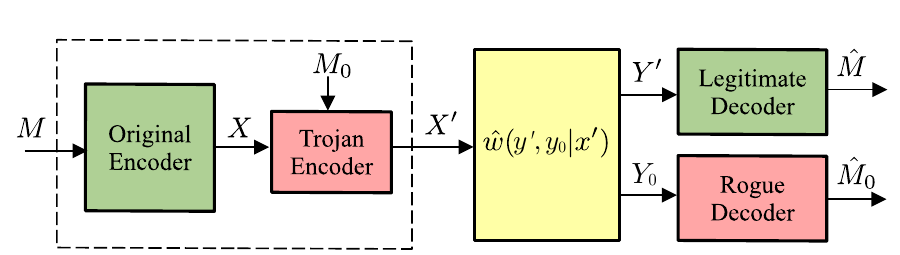}
    \caption{Trojan system model.} 
    \label{fig:2_Trojan}
    \vspace{-0.5em}
\end{figure}

This model differs from a standard broadcast channel and from standard
mismatched decoding in essential ways. A broadcast-channel transmitter is
typically designed jointly for multiple receivers. In the Transmitter Trojan
channel, the legitimate stream is fixed by an existing point-to-point
codebook, while the rogue stream is embedded through a constrained
symbol-by-symbol perturbation of the legitimate codeword. In standard
mismatched decoding, the channel-metric mismatch is typically exogenous to
the code construction. In the transmitter Trojan model, the mismatch is induced by the Trojan signaling
rule itself. Thus, the leakage rate to the rogue receiver, the achievable rate of the legitimate message, its decoding reliability at the legitimate receiver, and the statistical separation available for detecting the Trojan, are all influenced by the embedded perturbations of the Trojan transmitter.

The fixed-codebook nature of the Trojan perturbation adds a layer of complexity to the analysis. For the legitimate receiver, a fixed Trojan codebook induces a channel with memory; to obtain single-letter mismatched decoding rates, we use  additional randomization and channel resolvability. For reliability exponents, the corresponding resolvability approximation must be controlled with an explicitly bounded exponential rate, so that the approximation error does not dominate the decoding error exponent.
At the rogue receiver, the exponent behavior depends on whether the legitimate codeword is decodable, whether it is resolvable as interference, or whether neither property is available. Detection also depends on the receiver's statistical knowledge: likelihood-ratio tests are natural when the ensemble-averaged induced channel is known, while goodness-of-fit tests against the nominal channel are used when no statistical information about the Trojan-induced channel is available. Detectability is analyzed as an induced performance tradeoff, rather than imposed as an undetectability constraint on the Trojan code, reflecting the transmitter Trojan mechanisms reported in the literature. Accordingly, the model does not impose covert-communication or secrecy constraints, which would correspond to distinct operational requirements beyond the transmitter-Trojan model considered here.

To summarize, the main contributions of this paper are as follows:
\begin{itemize}
\item
We introduce the Transmitter Trojan channel as a model for symbol-by-symbol malicious perturbations embedded in a legitimate transmitter. The model induces a broadcast-like communication problem with one-sided decoder mismatch: the legitimate receiver continues to use the metric designed for the uncontaminated channel, while the rogue receiver is matched to the Trojan channel and the Trojan encoder.

\item
We derive achievable rogue message rates
and achievable rates for the legitimate message under the Trojan-induced mismatch. The rate analysis combines a Marton construction with mismatched-decoding tools, including generalized mutual information (GMI) and LM rates. Channel resolvability is used to connect fixed-codebook Trojan perturbations to ensemble-averaged induced single-letter channels.

\item
We provide a genie-aided converse that gives a rogue leakage rate outer bound and is tight when the rogue receiver can decode the legitimate transmission. 

\item
We develop reliability-exponent bounds for the legitimate and rogue receivers. For the legitimate receiver, memoryless mismatched exponents are transferred to operational fixed-codebook lower bounds through an exponential resolvability argument. For the rogue receiver, we distinguish decodable, resolvable, and intermediate regimes, yielding operational exponents in the first two regimes and averaged-interference benchmarks in the intermediate regime.

\item
We analyze Trojan detection at the legitimate receiver under different statistical-knowledge assumptions. When the ensemble-averaged induced channel is known, we derive likelihood-ratio detection exponents. When only the nominal channel is known, we analyze detector-specific goodness-of-fit exponents.

\item
We specialize the rate, reliability exponent, and detection exponent expressions to the binary symmetric channel (BSC) and additive white Gaussian noise (AWGN) channel models and provide numerical results illustrating the tradeoffs among rogue message rate, legitimate message rate, decoding reliability, and Trojan detectability.
\end{itemize}

\textit{Notations:} All logarithms have base $e$ in this paper unless otherwise specified. Random variables are represented by uppercase letters, and their realizations by lowercase letters.

\section{The Transmitter Trojan Model}
\label{sec:sys_model}

We start with a baseline uncontaminated system, shown in
Fig.~\ref{fig:1_Legitimate}. 
The uncontaminated channel input is denoted $X$ and the output is denoted $Y$. The legitimate message is
$M\in[1:e^{nR}]$ and the legitimate codebook is ${\mathcal C}$, which is generated i.i.d. according to $p_X(x)$. The channel to the legitimate receiver has
transition law $w(y|x)$, and the legitimate decoder 
produces the estimate $\hat M$.

The Trojan channel model consists of a malicious modification that perturbs $X^n$ prior to transmission, as well as the introduction of a {\em rogue receiver} (see Fig.~\ref{fig:2_Trojan}).
The {\em rogue message} is denoted $M_0 \in [1,e^{nR_0}]$. 
The {\em Trojan encoder} perturbs the legitimate codeword $X^n$ based on the rogue message $M_0$, producing the
emitted sequence $X'^n$. 
The physical channel is modeled as a memoryless broadcast channel from the
emitted symbol $X'$ to the legitimate and rogue observations $Y'$ and $Y_0$,
respectively, with transition law $\widehat w(y',y_0|x')$.
We assume that the physical channel from the emitted symbol to the legitimate
receiver is unaffected by the Trojan modification or by the presence of the
rogue receiver. Accordingly, the legitimate marginal of the broadcast channel
satisfies
\begin{align}
\widehat w_{Y'|X'}(y'|x')=w(y'|x').
\end{align}
The rogue marginal is denoted by
\begin{align}
w_0(y_0|x')\coloneqq \widehat w_{Y_0|X'}(y_0|x').
\end{align}

\section{Coding Strategy \& Achievable Rates} \label{sec:Achievable_Rates}

Let  $\mathsf C_0=\{V^n(m_0)\}_{m_0=1}^{e^{nR_0}}$ denote a random {\em Trojan codebook}, where the codewords are mutually independent and each has distribution $\prod_{i=1}^n p_V(v_i)$. A realization of $\mathsf C_0$ is denoted by $\mathcal C_0$. For rogue message $M_0$, the Trojan encoder combines $v^n(M_0)$ with the legitimate codeword $X^n$ in a symbol-by-symbol manner to generate the emitted symbols:
\begin{align}
X'_i=h(X_i,v_i(M_0)), \qquad i=1,\ldots,n.
\label{eq:TrojanEncoder2}
\end{align}
This coding strategy is consistent with the Trojans reported in the physical layer of communication systems and their practical constraints~\cite{liu2016silicon,subraman2019demonstrating,subramani2020amplitude,dhole2024wireless}.

\begin{figure}[t!]
    \centering
\includegraphics[scale=0.75]{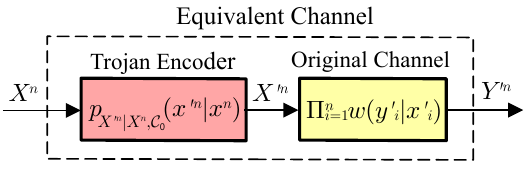}
    \caption{Equivalent channel to the legitimate receiver
    induced by the Trojan perturbation.}
    \label{fig:3_degraded}
    \vspace{-1em}
\end{figure}

To calculate achievable rates for the legitimate message and rogue message under the Trojan model,
we utilize the independent version of Marton's inner bound~\cite{marton1979coding}, originally introduced by Cover \cite{cover2003achievable} and van der Meulen \cite{van2003random}, in part because it includes a symbol-by-symbol mapper in its construction that is particularly applicable to the natural constraints in a Trojan model. 
For an \emph{i.i.d.} broadcast channel with outputs $Y',Y_0$, input $X'$, and Marton auxiliary variables $V$ and $X$, this inner bound is described by:
\begin{align}
R &\le I(X;Y'), \label{eq:MartonRate1}\\
R_0 &\le I(V;Y_0), \label{eq:MartonRate2}
\end{align}
for some distribution $p_X(x)p_V(v)$ and $X'=h(X,V)$ for some $h(\cdot,\cdot)$.

The legitimate decoder uses the metric designed for the uncontaminated channel $w(y'|x)$, whereas the actual conditional law from $X^n$ to $Y'^n$ is altered by the Trojan encoder. This places the legitimate receiver within the framework of mismatched decoding. The principal additional difficulty is that, for a fixed Trojan codebook $\mathcal C_0$, the induced channel ${\widecheck u}_{\mathcal C_0}(y'^n|x^n)$ generally has memory. We use the LM rate, originally developed for discrete memoryless channels using constant-composition random coding~\cite{hui1983fundamental,csiszar1981graph} and later extended to continuous-alphabet channels using cost-constrained random coding~\cite{ganti2000mismatched}. We also use the GMI obtained through i.i.d. random coding~\cite{merhav1994information}, which extends directly to continuous-alphabet channels~\cite{scarlett2020information}. These rates can also be generalized to channels with memory as shown in~\cite{ganti2000mismatched}. The corresponding mismatched-rate expressions for the Trojan model are given in Appendix~\ref{Appendix-MismatchedRates}.

We then modify the Marton rate~\eqref{eq:MartonRate1} to account for the mismatch in the manner stated above, arriving at the achievable rate of the Trojan-infested system. 
Assuming that $M_0$ is uniformly distributed, a fixed realization $\mathcal C_0$ induces, for each legitimate codeword $x^n$, the conditional distribution
\begin{align}
p_{X'^n|X^n,{\mathcal C}_0}(x'^n|x^n)=\frac{1}{e^{nR_0}}\sum_{v^n \in \mathcal{C}_0}\prod_{i=1}^{n}\mathbbm{1}_{\{x'_i=h(x_i,v_i)\}}.    
\end{align}
Thus characterized, the equivalent channel from $X^n$ to $Y'^n$ will obey:
\begin{align}\nonumber
{\widecheck u}_{{\mathcal C}_0}(y'^n|x^n)&
=\sum_{x'^n}p_{X'^n|X^n,{\mathcal C}_0}(x'^n|x^n)\prod_{i=1}^{n}w(y'_i|x'_i)\\\nonumber
&=\sum_{x'^n}\sum_{v^n\in \mathcal{C}_0}\frac{1}{e^{nR_0}}\prod_{i=1}^{n}[w(y'_i|x'_i)\mathbbm{1}_{\{x'_i=h(x_i,v_i)\}}]
\\\label{dist_u_memory}
&=\frac{1}{e^{nR_0}}\sum_{v^n \in \mathcal{C}_0}\prod_{i=1}^{n}w(y'_i|h(x_i,v_i)).
\end{align}

\begin{lemma}[Fixed-codebook achievable rate]
\label{lemma_Marton}
The Trojan-infested system characterized by the channel law $\widehat w(y',y_0|x')$ supports the rogue message rate:
\begin{align}
R_0 &\le I(V;Y_0),
\end{align}
where $x'=h(x,v)$ for some $h(\cdot,\cdot)$, under the distribution $p_X^*(x)p_V(v)$, where $p_X^*$ is the distribution of the uncontaminated codebook. With the Trojan, the legitimate receiver has the achievable rate:
\begin{align}
R &\le \lim_{n \rightarrow \infty}\frac{1}{n}I_{\textrm{M},\mathcal{C}_0}(X^n;Y'^n),\label{eq:MartonRate1-mismatch1}
\end{align}
where $I_{\textrm{M},\mathcal{C}_0}$ stands for the mismatched rate for a decoder designed for an i.i.d. channel $\prod_{i=1}^{n}w(y'_i|x_i)$ but operating under the channel ${\widecheck u}_{{\mathcal C}_0}(y'^n|x^n)$ shown in Eq.~\eqref{dist_u_memory}.
\end{lemma}
\begin{proof}
Please see Appendix~\ref{Appendix_Marton_Mismatch1}.
\end{proof}

The main usefulness of this result is to characterize the achievable rate $R_0$ for the rogue receiver, which highlights the severity of the information leakage. The rate $R$ for the legitimate receiver highlights the rate loss for the legitimate receiver compared with the uncontaminated case, for which in the most general scenario, Lemma~\ref{lemma_Marton} only provides a multi-letter characterization. 
To obtain a single-letter legitimate message rate, the fixed-codebook channel
${\widecheck u}_{{\mathcal C}_0}(y'^n|x^n)$ in~\eqref{dist_u_memory} must be approximated with a memoryless
conditional distribution. 
We do this by introducing an additional randomization index into the 
Trojan codebook.

\begin{theorem}[Single-letter achievable rates under channel resolvability]
\label{thm-singleletter-resolvability}
Consider the Trojan-infested system with finite alphabets and channel law
$\widehat w(y',y_0|x')$. Fix an input distribution $p_X^*(x)$, a Trojan-codeword
distribution $p_V(v)$, and a symbol-by-symbol mapper $x'=h(x,v)$. 
Define
the ensemble-averaged induced channel
\begin{align}\label{eq:induced_legitimate_channel}
u(y'|x)
\coloneqq
\sum_v p_V(v) w(y'|h(x,v)).
\end{align}
Assume that the legitimate decoder uses the metric designed for the uncontaminated
channel $w(y'|x)$, and that this metric is strictly positive on its support.

Suppose the Trojan encoder uses an additional randomization index
of rate $R'$, also known
to the rogue receiver, and that
\begin{align}
R_0+R' > I(V;Y'|X)
\end{align}
under the joint distribution
\begin{align}
p_X^*(x)p_V(v)\mathbbm{1}_{\{x'=h(x,v)\}}\widehat w(y',y_0|x').
\end{align}
Then, for every $\epsilon>0$, there exists a sequence of 
Trojan codebooks such that the legitimate receiver supports every rate
\begin{align}
R < I_{\rm M}(X;Y')-\epsilon,
\end{align}
where $I_{\rm M}(X;Y')$ denotes either the GMI or LM mismatched achievable
rate evaluated under the 
ensemble-averaged induced channel $u(y'|x)$ and decoding
metric $w(y'|x)$.

The rogue receiver supports every rate
\begin{align}
R_0 < I(V;Y_0)
\end{align}
for the independent-decoding construction.
\end{theorem}
\begin{proof}
Please see Appendix~\ref{Appendix_Marton_Mismatch2}.
\end{proof}

\begin{corollary}[Previous messages as resolvability randomization]
\label{corollary_previous_messages}
If the additional randomization index is supplied by $K$ previously decoded
rogue messages, then $R'=KR_0$ after initialization and with $K$-block delay.
In this case the resolvability condition becomes
\begin{align}
(K+1)R_0 > I(V;Y'|X).
\end{align}
Together with $R_0<I(V;Y_0)$.
Equivalently, after an arbitrarily small rate backoff, it is sufficient that
\begin{align}
I(V;Y'|X) < (K+1)I(V;Y_0).
\end{align}
\end{corollary}

Under this condition, the legitimate
receiver can be analyzed using the 
ensemble-averaged induced channel $u(y'|x)$ in \eqref{eq:induced_legitimate_channel},
while the rogue receiver continues to decode the rogue message through the corresponding random-coding construction.

When the rogue receiver is less noisy than the legitimate receiver, it can
decode the legitimate codeword and then decode the rogue message over the
conditional channel from $V$ to $Y_0$ given $X$. This gives the following
refinement of Theorem~\ref{thm-singleletter-resolvability}.

\begin{proposition}[Less-noisy rogue receiver]
\label{Proposition_superposition}
Consider the Trojan-infested system with broadcast channel law
$\widehat w(y',y_0|x')$. Suppose the rogue receiver is less noisy than the legitimate
receiver, i.e., for every $p(u,x')$,
\begin{align}
I(U;Y')\le I(U;Y_0).
\end{align}
Fix $p_X^*(x)$, $p_V(v)$, and $x'=h(x,v)$, and assume the randomization
condition in Theorem~\ref{thm-singleletter-resolvability} is satisfied. Then, for every
$\epsilon>0$, all rates satisfying
\begin{align}
R_0 &< I(V;Y_0|X)-\epsilon,\\
R &< I_{\rm M}(X;Y')-\epsilon
\end{align}
are achievable, where $I_{\rm M}(X;Y')$ is the GMI or LM mismatched
achievable rate under the ensemble-averaged induced channel $u(y'|x)$ and decoding metric
$w(y'|x)$.
\end{proposition}
\begin{proof}
Please see Appendix~\ref{Appendix_less_noisy_rogue}.
\end{proof}

\begin{proposition}[Genie-aided outer bound]
\label{prop:genie_outer_bound}
Consider any sequence of length-$n$ codes satisfying the symbol-wise Trojan
model
\begin{align}
X'_i=h(X_i,V_i),
\end{align}
over a memoryless broadcast channel $\widehat w(y',y_0|x')$. Suppose $M$ and $M_0$  are decoded by legitimate and rogue decoders, respectively, with
vanishing error probabilities. Then every achievable rate pair $(R,R_0)$
satisfies
\begin{align}
R
&\le I(X';Y'|Q),\label{eq:legitimate-rate}\\
R_0
&\le I(V;Y_0|X,Q),\label{eq:rogue-rate}\\
R+R_0
&\le I(X,V;Y',Y_0|Q),\label{eq:sumrate}
\end{align}
for some distribution of the form
\begin{align}
p(q)p(x,v|q)\mathbbm{1}_{\{x'=h(x,v)\}}\widehat w(y',y_0|x').
\end{align}
Here $Q$ is a time-sharing random variable independent of all messages. If
the induced code distribution is stationary and memoryless, $Q$ may be
omitted.
\end{proposition}
\begin{proof}
Please see Appendix~\ref{Appendix_genie_aided_outer_bound}.
\end{proof}

The bound is expressed in the variables of the constrained Trojan embedding.
In particular, the genie-aided constraint
$R_0\le I(V;Y_0|X,Q)$ highlights the rogue message rate limits when the rogue receiver can decode the legitimate transmission. The remaining inequalities provide matched
single-letter ceilings for the legitimate and joint rates under the
symbol-wise embedding constraint.

\section{Reliability and Trojan Detection Error Exponents}
\label{sec:Reliability}

This section studies reliability exponents and Trojan detection exponents. The reliability exponents describe exponential decay rates of decoding error probabilities, while the detection exponents describe exponential decay rates of false-alarm and missed-detection probabilities under specified knowledge assumptions.

Several exponent expressions are first evaluated for ensemble-averaged induced channels. For the actual fixed-codebook Trojan system, their operational
interpretation requires the relevant exponential resolvability or
interference-resolvability approximation. Unless otherwise stated, the
discrete-alphabet formulas are presented explicitly; the Gaussian
specialization is evaluated separately.

\subsection{Rogue Error Exponent}

For the rogue receiver, the reliability analysis separates into three
regimes. If the legitimate codeword can be decoded at the rogue receiver,
then the rogue message can be decoded according to the metric for the known-state channel
$p_{Y_0|V,X}$. If the legitimate codeword is not decoded but the legitimate
codebook randomizes the rogue observation sufficiently, then the operational error exponents can be derived according to an averaged
channel from $V$ to $Y_0$
through an exponent-transfer
argument. In the remaining intermediate regime, where the legitimate codeword
is neither decoded nor known to be resolvable at the rogue receiver, the
averaged-channel exponent is only an ensemble or averaged-interference
benchmark. 

Proposition~\ref{proposition_rogue_interference_resolvability} gives an operational condition for the
legitimate-resolvable regime and Proposition~\ref{proposition_exponent_transfer}
transfers the averaged-channel exponent to the operational
exponent for the actual fixed-legitimate-codebook exponents when the resolvability approximation holds with exponential
accuracy.

\begin{proposition}[Interference resolvability at the rogue receiver]
\label{proposition_rogue_interference_resolvability}
Fix a codeword $v^n$ and 
let $\mathsf C=\{X^n(m)\}_{m=1}^{e^{nR}}$ be a random codebook
with mutually independent codewords each according to
$\prod_{i=1}^n p_X^*(x_i)$. Define
\begin{align}
\widecheck{u}_{0,\mathsf C}(y_0^n|v^n)
\coloneqq
\frac{1}{e^{nR}}
\sum_{m=1}^{e^{nR}}
\prod_{i=1}^n
w_0(y_{0i}|h(X_i(m),v_i)).
\end{align}
Also define the averaged-interference channel 
\begin{align}
\label{eq:rogue_interference_averaged_channel}
u_0(y_0|v)
\coloneqq
\sum_x p_X^*(x)w_0(y_0|h(x,v)),
\end{align}
and let $u_0^n(y_0^n|v^n)\coloneqq \prod_{i=1}^{n}u_0(y_{0,i}|v_i)$. If
\begin{align}
R>I(X;Y_0|V),
\end{align}
then for every typical $v^n$ and for some constants
$\delta_1,\delta_2>0$,
\begin{align}
\mathbb P_{\mathsf C}
\left(
\left\|
\widecheck{u}_{0,\mathsf C}(\cdot|v^n)
-
u_0^n(\cdot|v^n)
\right\|_{\rm TV}
>
e^{-n\delta_1}
\right)
\le
e^{-e^{n\delta_2}}.
\end{align}
Consequently, with probability approaching one over the legitimate codebook
ensemble, the averaged-channel approximation holds exponentially for all
codewords $v^n$ of typical type.
\end{proposition}
\begin{proof}
For fixed $v^n$, the legitimate codebook acts as a random codebook driving
the virtual channel
\begin{align}
p_{Y_0|X,V}(y_0|x,v)=w_0(y_0|h(x,v)).
\end{align}
This is the conditional channel-resolvability problem with input $X$ and
output $Y_0$, conditioned on $V=v$. The same 
resolvability argument used in Appendix~\ref{conditional_resolvability} gives an
exponentially small total-variation approximation whenever
$R>I(X;Y_0|V)$. The double-exponential failure probability allows a union
bound over exponentially many codewords $v^n$ of typical type.
\end{proof}

\begin{proposition}[Exponent transfer under exponential channel approximation]
\label{proposition_exponent_transfer}
Fix a legitimate codebook $\mathcal C$ and a Trojan codebook  $\mathcal C_0$.
Suppose that, for every codeword
$v^n(m_0)\in\mathcal C_0$,
\begin{align}
\left\|
\widecheck u_{0,\mathcal C}(\cdot|v^n(m_0))
-
u_0^n(\cdot|v^n(m_0))
\right\|_{\rm TV}
\le e^{-nE_{\rm int}}
\label{eq:exp_tv_condition_rogue}
\end{align}
for some $E_{\rm int}>0$. Let
$P_{e,\widecheck u_0}^{(n)} (\mathcal C,\mathcal C_0)$ 
denote the average rogue
decoding error probability under the fixed-legitimate-codebook channel
$\widecheck u_{0,\mathcal C}$, and let $P_{e,u_0}^{(n)}(\mathcal C_0)$ 
denote the
corresponding error probability under the product channel 
$u_0^n$ 
using the same rogue decoder. 
Then
\begin{align}
P_{e,\widecheck u_0}^{(n)}(\mathcal C,\mathcal C_0)
\le
P_{e,u_0}^{(n)}(\mathcal C_0)+e^{-nE_{\rm int}}.
\label{eq:error_transfer_bound_rogue}
\end{align}
Consequently, if
\begin{align}\label{eq:error_bound_rogue}
P_{e,u_0}^{(n)}(\mathcal C_0)\le e^{-n(E_{0,\rm av}(R_0)-o(1))},
\end{align}
then
\begin{align}
P_{e,\widecheck u_0}^{(n)}(\mathcal C,\mathcal C_0)
\le
e^{-n(\min\{E_{0,\rm av}(R_0),E_{\rm int}\}-o(1))}.
\end{align}
Thus the operational rogue reliability exponent satisfies
\begin{align}
E_{\rm op}(R_0)\ge \min\{E_{0,\rm av}(R_0),E_{\rm int}\}.
\end{align}
\end{proposition}

\begin{proof}
For each message $m_0$, let $\mathcal E_{m_0}$ be the decoding error region under
the fixed rogue decoder. By the definition of total variation,
\begin{align}
\widecheck u_{0,\mathcal C}(\mathcal E_{m_0}|v^n(m_0))
&\le
u_0^n(\mathcal E_{m_0}|v^n(m_0))
\twocolbreak \includeonetwocol{}{\quad}
+\left\|
\widecheck u_{0,\mathcal C}(\cdot|v^n(m_0))
-
u_0^n(\cdot|v^n(m_0))
\right\|_{\rm TV}.
\end{align}
Using~\eqref{eq:exp_tv_condition_rogue} and averaging over messages gives
\eqref{eq:error_transfer_bound_rogue}. 
Combining it with \eqref{eq:error_bound_rogue}
gives
\begin{align}
P_{e,\widecheck u_0}^{(n)}(\mathcal C,\mathcal C_0)
\le
e^{-n(E_{0,\rm av}(R_0)-o(1))}
+
e^{-nE_{\rm int}}.
\end{align}
Since the sum of two exponentially decaying terms is dominated by the
slower-decaying term, we have
\begin{align}
P_{e,\widecheck u_0}^{(n)}(\mathcal C,\mathcal C_0)
\le
e^{-n(\min\{E_{0,\rm av}(R_0),E_{\rm int}\}-o(1))}.
\end{align}
The lower bound on $E_{0,\rm op}(R_0)$ then follows immediately.
\end{proof}

If the legitimate codeword is decodable at the rogue receiver, then after decoding $X^n$, the rogue message is decoded over the known-state channel
\begin{align}
p_{Y_0|V,X}(y_0|v,x)
=
w_0(y_0|h(x,v)).
\end{align}
For a typical legitimate codeword with empirical distribution $p_X^*$, the
known-state random-coding exponent is \cite{polyanskiy2025information}
\begin{align}
E_{0,{\rm ks}}(R_0)
=
\max_{\rho\in[0,1]}
\left\{
\Psi_{0,{\rm ks}}(\rho)-\rho R_0
\right\},
\end{align}
where
\begin{align}
&\Psi_{0,{\rm ks}}(\rho)
\twocolbreak \includeonetwocol{}{\quad}
=
-\sum_x p_X^*(x)
\log
\sum_{y_0}
\left(
\sum_v p_V(v)
w_0(y_0|h(x,v))^{\frac{1}{1+\rho}}
\right)^{1+\rho}.
\label{eq:rogue_known_state_E0}
\end{align}
The corresponding known-state expurgated exponent is \cite{polyanskiy2025information}
\begin{align}
E_{0,{\rm ks},{\rm ex}}(R_0)
=
\sup_{\rho\ge 1}
\left\{
\Psi_{0,{\rm ks},{\rm ex}}(\rho)-\rho R_0
\right\},
\end{align}
where
\begin{align}
\Psi_{0,{\rm ks},{\rm ex}}(\rho)
&=
-\rho
\sum_x p_X^*(x)
\log
\sum_{v,\bar v}
p_V(v)p_V(\bar v) \twocolbreak \includeonetwocol{}{\quad\times}
\left[
\sum_{y_0}
\sqrt{
w_0(y_0|h(x,v))
w_0(y_0|h(x,\bar v))
}
\right]^{\frac{1}{\rho}} .
\label{eq:rogue_known_state_ex}
\end{align}

In the legitimate-resolvable regime, where the rogue receiver does not decode
the legitimate codeword but the condition in
Proposition~\ref{proposition_rogue_interference_resolvability} holds, the legitimate
signal can be averaged out on the exponential scale. 
The averaged channel from $V$ to $Y_0$ is given in \eqref{eq:rogue_interference_averaged_channel}.
For this averaged channel,
the memoryless random-coding exponent is
\begin{align}
E_{0,{\rm av}}(R_0)
=
\max_{\rho\in[0,1]}
\left\{
\Psi_{0,{\rm av}}(\rho)-\rho R_0
\right\},
\end{align}
where
\begin{align}
\Psi_{0,{\rm av}}(\rho)
=
-\log
\sum_{y_0}
\left(
\sum_v p_V(v)
p_{Y_0|V}(y_0|v)^{\frac{1}{1+\rho}}
\right)^{1+\rho}.
\end{align}
The corresponding averaged-channel expurgated exponent is
\begin{align}
E_{0,{\rm av},{\rm ex}}(R_0)
=
\sup_{\rho\ge 1}
\left\{
\Psi_{0,{\rm av},{\rm ex}}(\rho)-\rho R_0
\right\},
\end{align}
where
\begin{align}
\Psi_{0,{\rm av},{\rm ex}}(\rho)
&=
-\rho
\log
\sum_{v,\bar v}
p_V(v)p_V(\bar v)
\twocolbreak \includeonetwocol{}{\quad\times}
\left[
\sum_{y_0}
\sqrt{
p_{Y_0|V}(y_0|v)
p_{Y_0|V}(y_0|\bar v)
}
\right]^{\frac{1}{\rho}} .
\end{align}

By Proposition~\ref{proposition_exponent_transfer}, the operational rogue exponents in the averaged-channel
regime satisfy
\begin{align}
E_{0,{\rm op}}(R_0)
&\ge
\min\{E_{0,{\rm av}}(R_0),E_{\rm int}\},\\
E_{0,{\rm op},{\rm ex}}(R_0)
&\ge
\min\{E_{0,{\rm av},{\rm ex}}(R_0),E_{\rm int}\}.
\end{align}

Outside the decodable and resolvable regimes, the formulas
$E_{0,{\rm av}}(R_0)$ and $E_{0,{\rm av},{\rm ex}}(R_0)$ should be interpreted
only as averaged-interference benchmarks.

\subsection{Legitimate Error Exponent}

For the legitimate receiver, we first evaluate the mismatched decoding exponent
for the ensemble-averaged induced
channel $u(y'|x)$ and the decoding metric matched
to the original channel $w(y'|x)$. 
Denote this memoryless mismatched exponent by $E_{\rm mm}(R)$. 
By the same
resolvability-based exponent-transfer argument used in Proposition~\ref{proposition_exponent_transfer},
the operational
exponent for the actual fixed-codebook induced channel is lower bounded by
\begin{align}
E_{\rm op}(R)\ge \min\{E_{\rm mm}(R),E_{\rm res}\},
\end{align}
where $E_{\rm res}$ is the exponent of the resolvability approximation. 
When $E_{\rm res}>E_{\rm mm}(R)$, the
memoryless mismatched exponent is recovered operationally.

An achievable random-coding exponent is given by
\cite{scarlett2014mismatched}
\begin{align}
E_{\rm mm}(R)
=
\max_{\rho\in[0,1]}
\left\{
\Psi(\rho)-\rho R
\right\},
\end{align}
For the i.i.d. ensemble,
\begin{align}
\Psi^{\rm iid}(\rho)
=
\sup_{s\ge 0}
&-\log
\sum_{x,y'} p_X^*(x)u(y'|x)
\twocolbreak \includeonetwocol{}{\quad\times}
\left(
\frac{
\sum_{\bar x} p_X^*(\bar x)\, w(y'|\bar x)^s
}{
w(y'|x)^s
}
\right)^{\rho}.\label{mismatched_exponent_iid}
\end{align}
For the constant-composition ensemble,
\begin{align}
\Psi^{\rm cc}(\rho)
=
\sup_{s\ge 0,\,a(\cdot)}
&-\log
\sum_{x,y'} p_X^*(x)u(y'|x)
\twocolbreak \includeonetwocol{}{\quad\times}
\left(
\frac{
\sum_{\bar x} p_X^*(\bar x)\, w(y'|\bar x)^s e^{a(\bar x)}
}{
w(y'|x)^s e^{a(x)}
}
\right)^{\rho}.\label{mismatched_exponent_cc}
\end{align}

Expurgated exponents for mismatched decoding can improve the low-rate
reliability behavior~\cite{scarlett2014expurgated}. An achievable
expurgated exponent is
\begin{align}
E_{{\rm mm},\rm ex}(R)
=
\sup_{\rho\ge 1}
\left\{
\Psi_{\rm ex}(\rho)-\rho R
\right\}.
\end{align}
For the i.i.d. ensemble,
\begin{align}
\Psi_{\rm ex}^{\rm iid}(\rho)
=
\sup_{s\ge 0}
&-\rho \log
\sum_{x,\bar x} p_X^*(x)p_X^*(\bar x)
\twocolbreak \includeonetwocol{}{\quad\times}
\left(
\sum_{y'} u(y'|x)
\left(
\frac{w(y'|\bar x)}{w(y'|x)}
\right)^s
\right)^{\frac{1}{\rho}}.\label{mismatched_exponent_iid_ex}
\end{align}
For the constant-composition ensemble,
\begin{align}
\Psi_{\rm ex}^{\rm cc}(\rho)
=
\sup_{s\ge 0,\,a(\cdot)}
&-\rho
\sum_x p_X^*(x)
\log
\sum_{\bar x} p_X^*(\bar x)
\twocolbreak \includeonetwocol{}{\quad\times}
\left(
\sum_{y'} u(y'|x)
\left(
\frac{w(y'|\bar x)}{w(y'|x)}
\right)^s
\frac{e^{a(\bar x)}}{e^{a(x)}}
\right)^{\frac{1}{\rho}}.\label{mismatched_exponent_cc_ex}
\end{align}

For continuous-input channels, the corresponding exponent expressions require
the appropriate cost-constrained ensembles and integrability conditions. In the Gaussian specialization in Section~\ref{Sec:Gaussian}, these exponents are evaluated directly rather than by a purely formal replacement of summations with integrals.

\subsection{Trojan Detection Error Exponent}

We study the asymptotic decay rates of the false-alarm and missed-detection probabilities for Trojan detection at the legitimate receiver, conditioned on a fixed legitimate codeword $x^n$. In this work, detectability is evaluated as a consequence of the Trojan strategy rather than imposed as a covertness constraint on the admissible Trojan codes. The results below are stated for the memoryless hypotheses induced by the nominal (original) channel $w(y'|x)$ and the ensemble-averaged induced channel $u(y'|x)$. For the fixed-codebook induced channel, the same interpretation requires the corresponding resolvability approximation. The transfer step is again based on total-variation comparison of event probabilities, so we state the single-letter detection exponents and omit the parallel fixed-codebook transfer details.

\subsubsection{
Ensemble-Averaged Induced Channel Knowledge}
Here, the legitimate receiver 
knows the ensemble-averaged induced channel
\begin{align}
u(y'|x)=\sum_v p_V(v)w(y'|h(x,v)).
\end{align}
For a fixed legitimate codeword $x^n$, the hypotheses are
\begin{equation}
\begin{aligned}
H_0 &: Y'^n\sim \prod_{i=1}^n w(y'_i|x_i),\\
H_1 &: Y'^n\sim \prod_{i=1}^n u(y'_i|x_i).
\end{aligned}
\end{equation}
The log-likelihood ratio is
\begin{align}
S_n=\sum_{i=1}^n \log \frac{u(Y'_i|x_i)}{w(Y'_i|x_i)}.
\end{align}
The false-alarm and missed-detection probabilities are
\begin{equation}
\begin{aligned}
P_{\rm FA} &= \mathbb{P}(S_n > \eta_n \mid H_0, x^n),\\
P_{\rm M} &= \mathbb{P}(S_n \le \eta_n \mid H_1, x^n),
\end{aligned}
\end{equation}
where $\eta_n$ is the decision threshold.

Under the Neyman--Pearson formulation, for any fixed $0<\epsilon<1$ and any
typical sequence $x^n$ with empirical distribution $p_X^*(x)$, the minimum
missed-detection probability among all tests satisfying
$P_{\rm FA}\le\epsilon$ decays exponentially with exponent
\begin{align}\label{Exponent_Stein}
E_{\rm M}^{\rm Stein}
&=
\sum_x p_X^*(x)
D\big(u(\cdot|x)\|w(\cdot|x)\big).
\end{align}
Thus, the Trojan is difficult to detect when the ensemble-averaged induced channel $u(y'|x)$ is
close to the original channel $w(y'|x)$ in KL divergence.

Alternatively, under a Bayesian formulation with prior probabilities
$p(H_0)$ and 
$p(H_1)$, the overall detection error probability is
\begin{align}
P_{{\rm e,det}}
=
p(H_0) P_{\rm FA}
+
p(H_1) P_{\rm M}.
\end{align}
For a fixed typical sequence $x^n$, the minimum Bayesian detection error
probability, corresponding to the MAP detector, decays exponentially with
Chernoff exponent
\begin{align}\label{Exponent_Chernoff}
E_{\rm e}^{\rm Ch}
=
-\min_{0\le \lambda\le 1}
\sum_x p_X^*(x)
\log
\sum_{y'} w(y'|x)^\lambda u(y'|x)^{1-\lambda}.
\end{align}

More generally, for a threshold sequence $\eta_n=n\tau$, the false-alarm and
missed-detection exponents of the likelihood-ratio test are
\begin{equation}
\begin{aligned}
E_{\rm FA}(\tau)
&=
\sup_{\theta\ge 0}
\left\{
\theta\tau
-
\Lambda_{\rm FA}(\theta)
\right\},\\\label{Exponent_Large_Deviation}
E_{\rm M}(\tau)
&=
\sup_{\theta\le 0}
\left\{
\theta\tau
-
\Lambda_{\rm M}(\theta)
\right\},
\end{aligned}
\end{equation}
where
\begin{equation}
\begin{aligned}
\Lambda_{\rm FA}(\theta)
&=
\sum_x p_X^*(x)
\log
\sum_{y'} w(y'|x)
\left(\frac{u(y'|x)}{w(y'|x)}\right)^\theta,\\\label{Lambda_Large_Deviation}
\Lambda_{\rm M}(\theta)
&=
\sum_x p_X^*(x)
\log
\sum_{y'} u(y'|x)
\left(\frac{u(y'|x)}{w(y'|x)}\right)^\theta.
\end{aligned}
\end{equation}

\subsubsection{Nominal-Channel-Only Knowledge}

When the legitimate receiver knows only the nominal channel $w(y'|x)$, it
cannot form the likelihood ratio involving $u(y'|x)$. In this case we use a
goodness-of-fit detector based on the nominal joint distribution
\begin{align}
P_0(x,y')=p_X^*(x)w(y'|x).
\end{align}
For a fixed legitimate codeword $x^n$, let $\widehat P_{X,Y'}$ denote the
empirical joint distribution of the pairs $(x_i,Y'_i)$. The detector accepts
$H_0$ when $\widehat P_{X,Y'}$ belongs to an acceptance region $\mathcal A$
chosen using only the nominal distribution $P_0$.

For example, a Pearson chi-squared goodness-of-fit detector uses the
acceptance region
\begin{align}
\mathcal A_{\chi^2}
=
\left\{
Q:
\sum_{x,y'}
\frac{\big(Q(x,y')-P_0(x,y')\big)^2}
{P_0(x,y')}
\le \tau
\right\},
\end{align}
where $\tau$ is chosen to satisfy the desired false-alarm constraint. More
generally, for any acceptance region $\mathcal A$ determined by the nominal
model, Sanov's theorem gives the detector-specific exponents \cite{dembo2009large}
\begin{equation}
\begin{aligned}
E_{\rm FA}
&=
\inf_{Q\notin\mathcal A}
D(Q\|P_0),\\\label{Exponent_Sanov}
E_{\rm M}
&=
\inf_{Q\in\mathcal A}
D(Q\|P_1),
\end{aligned}
\end{equation}
where
\begin{align}
P_1(x,y')=p_X^*(x)u(y'|x)
\end{align}
is the  
ensemble-averaged induced alternative used to evaluate the missed-detection
probability. Thus, in the 
nominal-channel-only knowledge
case, $u(y'|x)$ is not used
to construct the detector; it is used only to analyze the detector's
performance under the ensemble-averaged induced alternative.

\section{Binary Symmetric Channel}

Here, we consider BSCs described by
\begin{equation}
\begin{aligned}
Y'&=X'\oplus Z,\\
Y_0&=X' \oplus Z_0,
\end{aligned}
\end{equation}
where $Z\sim \mathrm{Bern}(p)$ and $Z_0\sim \mathrm{Bern}(p_0)$
While the original legitimate encoder may, in general, use a suboptimal input
distribution, we assume it is optimally designed for the original channel,
i.e., $X\sim \mathrm{Bern}(\frac12)$.

For the BSC specialization, we use the independent symbol-wise construction
from Section~\ref{sec:Achievable_Rates}. 
rogue codewords are generated i.i.d. according to $\mathrm{Bern}(\rho_0)$ and the Trojan emits
\begin{align}
X'=X\oplus V.
\end{align}
When the rogue receiver is less noisy than the legitimate receiver, the Trojan-infested BSC system is equivalent to the
physically degraded channel depicted in Fig.~\ref{fig:BSC}, where
$p_1=\frac{p-p_0}{1-2p_0}$. 
The required randomization is supplied as in
Theorem~\ref{thm-singleletter-resolvability}, either by an additional
randomization index or by previously decoded rogue messages as in
Corollary~\ref{corollary_previous_messages}.
Hence, the ensemble-averaged induced channel
is a BSC with crossover probability
\begin{align}
\widetilde p=\rho_0+p-2\rho_0p .
\end{align}

\begin{figure}[t!]
 \centering
 \includegraphics[scale=0.65]{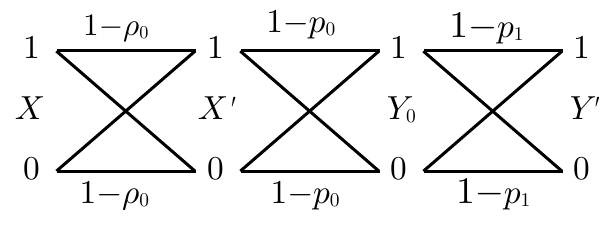}\\ \includegraphics[scale=0.65]{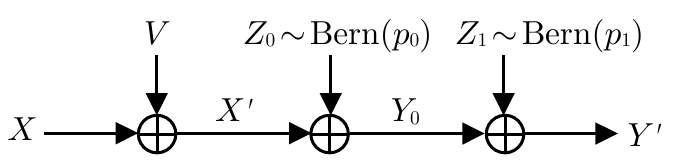}
    \caption{Trojan-infested binary symmetric channel.}
    \label{fig:BSC}
    \vspace{-0.5 em}
\end{figure}

To obtain the achievable legitimate message rate,
we evaluate the LM and GMI
mismatched rates. Let $\mathbf W$ and $\mathbf U$ denote the matrix representations of the
decoding metric and
the transition matrix of 
the ensemble-averaged induced channel, respectively:
\begin{align}
\mathbf{W}=
\begin{bmatrix}
1-p & p \\
p & 1-p
\end{bmatrix}, \quad
\mathbf{U}=
\begin{bmatrix}
1-\widetilde p & \widetilde p \\
\widetilde p & 1-\widetilde p
\end{bmatrix}.
\end{align}
\color{black}
For a binary channel, the LM rate is either equal to the matched rate or
zero \cite{scarlett2020information}, i.e.,
\begin{align}\label{LM_binary1}
I_{\rm LM}(X;Y')&=\begin{cases}
I(X;Y'), \quad 
\det(\mathbf{W})\det(\mathbf{U})\geq 0,
\\
0, \quad & \textrm{otherwise}.
\end{cases}
\end{align}
The GMI rate is the same as the LM rate for a BSC with a symmetric decoding
metric \cite{scarlett2020information} (see Appendix~\ref{Appendix_GMI_BSC}).
We find the achievable rates of a Trojan-infested BSC in
Lemma~\ref{lemma_binary}.

\begin{lemma}[BSC achievable rates]
\label{lemma_binary}
Using the GMI and LM rates, the achievable rates of a Trojan-infested BSC 
are obtained as all rates $(R_0,R)$ satisfying 
\begin{equation}
\begin{aligned}
R_0 &\leq 
\big[
H(\rho_0 + p_0 -2\rho_0p_0) - H (p_0)
\big]\mathbbm{1}_{\{p_0 \leq p\}},\\\label{eq:BSC_Rates}
R &\leq \big[1-
H(\widetilde p)
\big]\mathbbm{1}_{\{\rho_0 \leq \frac{1}{2}\}},
\end{aligned}
\end{equation}
for some $\rho_0 \in [0,1]$, where $H(a)\coloneqq -a \log a - (1-a) \log (1-a)$. 
 \end{lemma}
 \begin{proof} The proof is straightforward using 
Proposition~\ref{Proposition_superposition} and \eqref{LM_binary1}.
 \end{proof}

For the error exponent analysis, we restrict attention to the parameter
regimes for which the achievable rates in Lemma~\ref{lemma_binary} are
nonzero, i.e., $p_0\le p$ for the rogue receiver and
$\rho_0\le \frac12$ for the legitimate receiver. We also assume
$p,p_0\in[0,\frac12]$. In the BSC specialization, the averaged-interference
rogue channel is uninformative because $X\sim\mathrm{Bern}(\frac12)$ and
$Y_0=X\oplus V\oplus Z_0$, which gives $p_{Y_0|V}(y_0|v)=\frac12$ when
$X^n$ is not decoded. Hence the positive-rate rogue exponent below
corresponds to the decode-and-remove regime, where the rogue receiver first
decodes the legitimate codeword and then decodes $V^n$ over an effective
BSC with crossover probability $p_0$.

\begin{lemma}[BSC rogue error exponents]\label{lemma_BSC_exponent}
For a Trojan-infested BSC, an achievable random-coding error exponent for the rogue receiver, when
$p_0\le p$, is
\begin{align}
E_{0,\textrm{ks}}^{\textrm{BSC}}(R_0)
=
\max_{\rho\in[0,1]}
\left\{
\Psi_{0,\textrm{ks}}^{\textrm{BSC}}(\rho)-\rho R_0
\right\},
\end{align}
where
\begin{align}\label{psi_rogue}
\Psi_{0,\textrm{ks}}^{\textrm{BSC}}(\rho)
=
-\log \bigg[
&\bigg(
(1-\rho_0)(1-p_0)^{\frac{1}{1+\rho}}
+\rho_0 p_0^{\frac{1}{1+\rho}}
\bigg)^{1+\rho}
\twocolbreak
+
\bigg(
(1-\rho_0)p_0^{\frac{1}{1+\rho}}
+\rho_0 (1-p_0)^{\frac{1}{1+\rho}}
\bigg)^{1+\rho}
\bigg].
\end{align}

In addition, an achievable expurgated exponent is
\begin{align}
E_{0,\textrm{ks},\rm ex}^{\textrm{BSC}}(R_0)
=
\sup_{\rho\ge 1}
\left\{
\Psi_{0,\textrm{ks},\rm ex}^{\textrm{BSC}}(\rho)-\rho R_0
\right\},
\end{align}
where
\begin{align}
\Psi_{0,\textrm{ks},\rm ex}^{\textrm{BSC}}(\rho)
=
-\rho \log
&\bigg[
1 - 2\rho_0(1-\rho_0)
\twocolbreak \includeonetwocol{}{\quad\times}
\label{psi_rogue_ex}
\bigg(
1 - \Big(2\sqrt{p_0(1-p_0)}\Big)^{\frac{1}{\rho}}
\bigg)
\bigg].
\end{align}
\end{lemma}
\begin{proof}
When $p_0\le p$, the rogue receiver can first decode $X^n$ and remove it,
yielding an effective memoryless BSC with crossover probability $p_0$. 
After subtracting the decoded legitimate codeword, the effective channel from
$V$ to $Y_0$ is independent of the realization of $X$. Hence the type
averaging over $p_X^*$ in~\eqref{eq:rogue_known_state_E0} collapses to the
ordinary BSC exponent.
The
results then follow by straightforward evaluation of the corresponding random-coding and expurgated exponents for this effective channel.
\end{proof}

\begin{lemma}[BSC legitimate error exponents]
\label{lemma_BSC_legitimate_exponent}
For the legitimate receiver in a Trojan-infested BSC, when
$\rho_0\le \frac{1}{2}$, the i.i.d. and constant-composition mismatched
random-coding exponents for the 
ensemble-averaged induced BSC coincide. The
memoryless mismatched exponent is
\begin{align}
E_{\rm mm}^{\rm BSC}(R)
=
\max_{\rho\in[0,1]}
\left\{
\Psi^{\textrm{BSC}}(\rho)-\rho R
\right\},
\end{align}
where
\begin{align}\label{eq:lemma_BSC_E0}
\Psi^{\textrm{BSC}}(\rho)
=
\rho \log 2
-
(1+\rho)
\log
\left[
(1-\widetilde p)^{\frac{1}{1+\rho}}
+
\widetilde p^{\frac{1}{1+\rho}}
\right].
\end{align}

The corresponding memoryless mismatched expurgated exponent is
\begin{align}
E_{{\rm mm},{\rm ex}}^{\rm BSC}(R)
=
\sup_{\rho\ge 1}
\left\{
\Psi_{\rm ex}^{\textrm{BSC}}(\rho)-\rho R
\right\},
\end{align}
where
\begin{align}
\Psi_{\rm ex}^{\textrm{BSC}}(\rho)
=
-\rho \log
\left[
\frac{1}{2}
\left(
1+
\left(2\sqrt{\widetilde p(1-\widetilde p)}\right)^{\frac{1}{\rho}}
\right)
\right].
\end{align}

For the actual fixed-codebook induced 
channel, Proposition~\ref{proposition_exponent_transfer}
gives the operational lower bounds
\begin{align}
E_{\rm op}^{\rm BSC}(R)
&\ge
\min\{E_{\rm mm}^{\rm BSC}(R),E_{\rm res}\},\\
E_{{\rm op},{\rm ex}}^{\rm BSC}(R)
&\ge
\min\{E_{{\rm mm},{\rm ex}}^{\rm BSC}(R),E_{\rm res}\}.
\end{align}
When $\rho_0>\frac12$, the corresponding achievable rate and error exponents are zero.
\end{lemma}
\begin{proof}
The proofs of the memoryless mismatched exponents are provided in
Appendix~\ref{Appendix_BSC_error_exponent}. The operational bounds follow
from Proposition~\ref{proposition_exponent_transfer}.
\end{proof}

\begin{lemma}[BSC Trojan detection error exponents]\label{lemma_BSC_detection_exponents}
For Trojan detection at the legitimate receiver in a
Trojan-infested BSC, when the 
ensemble-averaged induced
channel is known, the Stein
missed-detection exponent is
\begin{align}
E_{\rm M}^{\rm Stein}
=
(1-\widetilde p)\log\frac{1-\widetilde p}{1-p}
+
\widetilde p\log\frac{\widetilde p}{p}.
\end{align}

For the Bayesian formulation, the Chernoff detection error exponent is
\begin{align}\label{Bayesian_Cher}
E_{\rm e}^{\rm Ch}
=
-\min_{0\le \lambda\le 1}
\log
\left[
(1-p)^\lambda(1-\widetilde p)^{1-\lambda}
+
p^\lambda \widetilde p^{1-\lambda}
\right].
\end{align}

Moreover, for a threshold sequence $\eta_n=n\tau$, the false-alarm and
missed-detection exponents of the likelihood-ratio test are
\begin{equation}
\begin{aligned}
E_{\rm FA}(\tau)
&=
\sup_{\theta\ge 0}
\left\{
\theta\tau-\Lambda_{\rm FA}(\theta)
\right\},\\
E_{\rm M}(\tau)
&=
\sup_{\theta\le 0}
\left\{
\theta\tau-\Lambda_{\rm M}(\theta)
\right\},
\end{aligned}
\end{equation}
where
\begin{equation}
\begin{aligned}
\Lambda_{\rm FA}(\theta)
&=
\log\left[
(1-p)
\left(\frac{1-\widetilde p}{1-p}\right)^\theta
+
p
\left(\frac{\widetilde p}{p}\right)^\theta
\right],\\
\Lambda_{\rm M}(\theta)
&=
\log\left[
(1-\widetilde p)
\left(\frac{1-\widetilde p}{1-p}\right)^\theta
+
\widetilde p
\left(\frac{\widetilde p}{p}\right)^\theta
\right].
\end{aligned}
\end{equation}

Under nominal-channel-only knowledge,
consider the empirical-crossover
goodness-of-fit detector
\begin{align}
\widehat p
=
\frac1n
\sum_{i=1}^n
\mathbbm{1}\{Y'_i\neq x_i\},
\end{align}
which accepts $H_0$ when
\begin{align}
\widehat p\in\mathcal A,
\qquad
\mathcal A=[p-\tau,p+\tau].
\end{align}
For this detector, the false-alarm and missed-detection exponents are
\begin{equation}
\begin{aligned}
E_{\rm FA}
&=
\inf_{\widehat p\notin \mathcal A}
D\big(\mathrm{Bern}(\widehat p)\|\mathrm{Bern}(p)\big),\\
E_{\rm M}
&=
\inf_{\widehat p\in \mathcal A}
D\big(\mathrm{Bern}(\widehat p)\|\mathrm{Bern}(\widetilde p)\big).
\end{aligned}
\end{equation}
\end{lemma}
\begin{proof}
Under induced-channel knowledge, 
the Stein exponent, the
Chernoff exponent, and the false-alarm and missed-detection exponents for
the likelihood-ratio threshold test follow from
\eqref{Exponent_Stein}, \eqref{Exponent_Chernoff}, and
\eqref{Exponent_Large_Deviation}, respectively, by evaluating the
corresponding expressions for the BSC transition laws $w(\cdot|x)$ and $u(\cdot|x)$ with crossover probabilities $p$ and $\widetilde{p}$, respectively.

Under nominal-channel-only knowledge, 
the detector is constructed only from
the nominal BSC crossover probability $p$. Since the empirical joint
distribution is fully characterized by the empirical crossover probability
$\widehat p$, the goodness-of-fit acceptance region reduces to the interval
$\mathcal A=[p-\tau,p+\tau]$. Applying the Sanov exponents in
\eqref{Exponent_Sanov} to
$P_0=\mathrm{Bern}(p)$ and to the ensemble-averaged induced alternative
$P_1=\mathrm{Bern}(\widetilde p)$ gives the stated detector-specific
exponents.
\end{proof}

\section{Gaussian Channel}
\label{Sec:Gaussian}
\begin{figure}[t!]
    \centering
\includegraphics[scale=0.65]{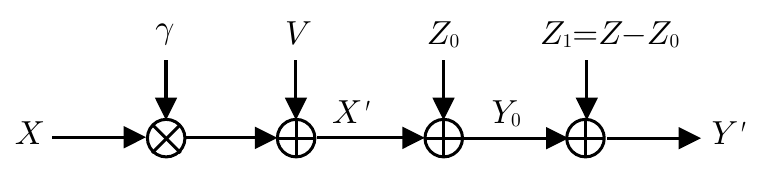}
    \caption{Trojan-infested Gaussian channel with $\sigma_0^2 \le \sigma^2$.} 
    \label{fig:5_degraded_Gaussian}
    \vspace{-1em}
\end{figure}

For the Trojan-infested Gaussian channel,
\begin{equation}
\begin{aligned}
Y'&=X'+Z,\\
Y_0&=X'+Z_0,
\end{aligned}
\end{equation}
where $Z_0 \sim \mathcal{N}(0,\sigma_0^2)$ and $Z \sim \mathcal{N}(0,\sigma^2)$. When the rogue receiver is less noisy than the legitimate receiver, i.e., $\sigma_0^2 \leq \sigma^2$, we have a degraded Gaussian channel depicted in Fig.~\ref{fig:5_degraded_Gaussian}.
There is a power constraint on the output of the transmitter as $\mathbb{E}[X'^2] \leq P$. 
We assume the original legitimate encoder is optimally designed for the original channel and operates at the boundary of the power constraint, i.e., $X\sim \mathcal{N}(0,P)$.
The Trojan codebook, being additive and independent, increases the transmit power. 
To obey the channel power constraint, the Gaussian Trojan modification consists of a scaling followed by superposition:
\begin{align}
X'=\sqrt{1-\alpha}\,X+V,
\end{align}
where $0\le \alpha \le 1$ and
$V\sim \mathcal N(0,\alpha P)$. The parameter
$\alpha$ controls the fraction of the transmit power that the contaminated transmitter diverts to the emission of the rogue message.
This is the Gaussian counterpart of the independent broadcast coding
construction used in Section~\ref{sec:Achievable_Rates}. 
Under the randomization condition in
Theorem~\ref{thm-singleletter-resolvability}, the fixed-codebook induced channel 
is approximated by the memoryless Gaussian channel
\begin{align}
Y'=\sqrt{1-\alpha} X+Z',
\end{align}
where $Z'=Z+V\sim\mathcal N(0,\sigma^2+\alpha P)$.

The general resolvability and exponent-transfer statements in
Sections~\ref{sec:Achievable_Rates} and~\ref{sec:Reliability} were stated for
finite alphabets. For the Gaussian specialization, the same role is played by
the standard continuous-alphabet resolvability approximation under
second-moment constraints. The expressions below are therefore evaluated
directly for the ensemble-averaged induced Gaussian channels, and their fixed-codebook
operational interpretation relies on the corresponding Gaussian
resolvability approximation.

In the following lemma, we obtain the achievable rates of a Trojan-infested Gaussian channel using the GMI and LM rates.

\begin{lemma}[Gaussian achievable rates]\label{LM_GMI_rates_G_lemma}
For a given $\alpha \in [0,1]$, the rogue achievable rate in a Gaussian
Trojan-infested system is 
\begin{align}
R_0 \le \begin{cases}
\frac{1}{2}\log\Big(1+\frac{\alpha P}{\sigma_0^2}\Big)
& \sigma_0^2 \le \sigma^2,
\\
\frac{1}{2}\log\Big(1+\frac{\alpha P}
{(1-\alpha)P+\sigma_0^2}\Big) & \sigma_0^2 > \sigma^2,
\end{cases}
\end{align}

Using the GMI rate, the legitimate achievable rate satisfies
\begin{align}
R &\le
\frac{1}{2}\log\Big(1+\frac{(1-\alpha)P}
{\sigma^2+\alpha P}\Big)
-
\frac{1}{2}\bigg[
\frac{
\sigma_{u'}^2+
(\beta_{u'}-\beta_{g}(s^*))^2\sigma_{y'}^2
}{
\sigma_{g}^2(s^*)
}
\twocolbreak \includeonetwocol{}{\quad}
+
\log\bigg(\frac{\sigma_{g}^2(s^*)}{\sigma_{u'}^2}\bigg)-1
\bigg],
\label{GMI_rate_G_R1}
\end{align}
where
$\sigma_{y'}^2=P+\sigma^2$,
$\sigma_{u'}^2=\frac{P(\alpha P+\sigma^2)}{P+\sigma^2}$,
$\beta_{u'}=\frac{P\sqrt{1-\alpha}}{P+\sigma^2}$,
$\sigma_{g}^2(s)=\frac{P\sigma^2}{sP+\sigma^2}$,
and
$\beta_{g}(s)=\frac{sP}{sP+\sigma^2}$.
Further, $s^*$ is obtained as
\begin{equation}\label{Gaussian_opts}
\begin{aligned}
s^*=
\tfrac{
-(3-4\sqrt{1-\alpha}+\frac{2\sigma^2}{P})
+
\sqrt{
(1+\frac{2\sigma^2}{P})^2
+
8(1-\sqrt{1-\alpha})(1+\frac{\sigma^2}{P})
}
}{
\frac{4P}{\sigma^2}(1-\sqrt{1-\alpha})+2
}.
\end{aligned}
\end{equation}

Using the LM rate, the legitimate achievable rate satisfies
\begin{align}\label{LM_rate_G_R1}
R
\le
\frac{1}{2}\log\Big(1+\frac{(1-\alpha)P}
{\sigma^2+\alpha P}\Big).
\end{align}
\end{lemma}

\begin{proof}
The rogue achievable rate is straightforward from the corresponding matched Gaussian channel. When $\sigma_0^2\le \sigma^2$, the rogue
receiver can first decode and remove the legitimate signal, yielding an AWGN channel with signal power $\alpha P$ and noise variance
$\sigma_0^2$. 
Otherwise, when the legitimate signal is not decoded and the corresponding
interference-resolvability approximation is used, the rogue receiver sees an
averaged Gaussian channel with effective noise variance
$(1-\alpha)P+\sigma_0^2$.

The proofs of the GMI and LM mismatched rates for the legitimate message are provided in 
Appendices~\ref{Appendix_GMI_Gaussian} and
\ref{Appendix_LM_Gaussian}, respectively.
\end{proof}
\color{black}

\color{black}
As observed in Lemma~\ref{LM_GMI_rates_G_lemma}, the achievable legitimate message
rate using the LM rate for the Gaussian channel is the same as the matched rate--an expected result, consistent with findings for the Gaussian channel with an unknown signal level~\cite{merhav1994information}. This is because, in such a channel, matched and mismatched maximum likelihood decoding are equivalent when the codewords have identical energy. In the LM rate, the use of cost-constrained random coding allows for generating codewords with identical energy, thereby achieving the matched rate.

\begin{lemma}[Gaussian rogue error exponents]\label{Gaussian_rogue_exp_lemma}
For a Gaussian Trojan-infested system, when $\sigma_0^2 \leq \sigma^2$, the legitimate
signal is decodable at the rogue receiver, and the corresponding known-state random coding exponent is
\begin{align}\label{E0_Gaussian}
E_{0,{\rm ks}}^{\rm G}(R_0)
=
\max_{\rho\in[0,1]}
\left\{
\Psi_{0,{\rm ks}}^{\rm G}(\rho)-\rho R_0
\right\},
\end{align}
where
\begin{align}
\Psi_{0,{\rm ks}}^{\rm G}(\rho)
=
\frac{\rho}{2}
\log\left(
1+\frac{\alpha P}{(1+\rho)\sigma_0^2}
\right).
\label{eq:Psi_rogue_G_ks}
\end{align}
The corresponding known-state expurgated exponent is
\begin{align}\label{E0_ex_Gaussian}
E_{0,{\rm ks},{\rm ex}}^{\rm G}(R_0)
=
\sup_{\rho\ge 1}
\left\{
\Psi_{0,{\rm ks},{\rm ex}}^{\rm G}(\rho)-\rho R_0
\right\},
\end{align}
where
\begin{align}
\Psi_{0,{\rm ks},{\rm ex}}^{\rm G}(\rho)
=
\frac{\rho}{2}
\log\left(
1+\frac{\alpha P}{2\rho\sigma_0^2}
\right).
\label{eq:Psi_rogue_ex_G_ks}
\end{align}

When $\sigma_0^2 > \sigma^2$, 
but the interference-resolvability condition in Proposition~\ref{proposition_rogue_interference_resolvability} holds,
then
the corresponding averaged-channel
random-coding exponent is
\begin{align}
E_{0,{\rm av}}^{\rm G}(R_0)
=
\max_{\rho\in[0,1]}
\left\{
\Psi_{0,{\rm av}}^{\rm G}(\rho)-\rho R_0
\right\},
\end{align}
where
\begin{align}
\Psi_{0,{\rm av}}^{\rm G}(\rho)
=
\frac{\rho}{2}
\log\left(
1+\frac{\alpha P}{(1+\rho)((1-\alpha)P+\sigma_0^2)}
\right).
\label{eq:Psi_rogue_G_av}
\end{align}
The corresponding averaged-channel expurgated exponent is
\begin{align}
E_{0,{\rm av},{\rm ex}}^{\rm G}(R_0)
=
\sup_{\rho\ge 1}
\left\{
\Psi_{0,{\rm av},{\rm ex}}^{\rm G}(\rho)-\rho R_0
\right\},
\end{align}
where
\begin{align}
\Psi_{0,{\rm av},{\rm ex}}^{\rm G}(\rho)
=
\frac{\rho}{2}
\log\left(
1+\frac{\alpha P}{2\rho((1-\alpha)P+\sigma_0^2)}
\right).
\label{eq:Psi_rogue_ex_G_av}
\end{align}
If the interference-resolvability approximation holds with exponent
$E_{\rm int}$, then the operational rogue exponents in the averaged-channel
regime satisfy
\begin{align}
E_{0,{\rm op}}^{\rm G}(R_0)
&\ge
\min\{E_{0,{\rm av}}^{\rm G}(R_0),E_{\rm int}\},\\
E_{0,{\rm op},{\rm ex}}^{\rm G}(R_0)
&\ge
\min\{E_{0,{\rm av},{\rm ex}}^{\rm G}(R_0),E_{\rm int}\}.
\end{align}
Outside the legitimate-decodable and legitimate-resolvable regimes, the
averaged-channel Gaussian exponents should be interpreted only as
averaged-interference benchmarks.
\end{lemma}
\begin{proof}
When $\sigma_0^2\le \sigma^2$, the legitimate signal is decoded at the rogue receiver and subtracted. This creates an equivalent AWGN channel with power constraint
$\alpha P$ and noise variance $\sigma_0^2$. Evaluating the matched
Gaussian random-coding and expurgated exponents gives
\eqref{eq:Psi_rogue_G_ks} and~\eqref{eq:Psi_rogue_ex_G_ks}.

When $\sigma_0^2>\sigma^2$, but the condition in
Proposition~\ref{proposition_rogue_interference_resolvability} holds, the conditional distribution $p(y_0|v)$ is resolvable on the exponential scale. The equivalent channel is once again an
AWGN channel with power constraint $\alpha P$, except its equivalent noise
variance is $(1-\alpha) P +\sigma_0^2$. Re-evaluating matched Gaussian
random-coding and expurgated exponents gives
\eqref{eq:Psi_rogue_G_av} and~\eqref{eq:Psi_rogue_ex_G_av}. The operational
bounds follow from
Proposition~\ref{proposition_exponent_transfer}.
\end{proof}

\begin{lemma}[Gaussian legitimate error exponents]\label{Gaussian_legitimate_exp_lemma}
For the legitimate receiver in a Gaussian Trojan-infested system, the
memoryless mismatched random-coding exponent is
\begin{align}\label{eq:E_Gaussian}
E_{\rm mm}^{\rm G}(R)
=
\max_{\rho\in[0,1]}
\left\{
\Psi(\rho)-\rho R
\right\}.
\end{align}
For the i.i.d. ensemble,
\begin{align}\label{eq:Psi_Leg_iid_G}
\Psi^{\rm iid}(\rho)
=\sup_{\substack{s\ge 0:\\ \Delta(s,\rho)>0}}
\frac{1}{2}
\log
\left[
\left(\frac{\sigma^2+sP}{\sigma^2}\right)^\rho
\frac{\Delta(s,\rho)}
{\sigma^2(\sigma^2+sP)}
\right],
\end{align}
where
\begin{align}\label{eq:Delta_iid_G}
\Delta(s,\rho)
&=
\sigma^4+\sigma^2Ps\big[1+\rho(2\sqrt{1-\alpha}-1)\big] \twocolbreak \includeonetwocol{}{\quad}
-Ps^2\rho
\left[
(\sigma^2+\alpha P)(1+\rho)+P(1-\sqrt{1-\alpha})^2
\right].
\end{align}
For the cost-constrained ensemble,
\begin{align}
\Psi^{\rm cc}(\rho)
=
\frac{\rho}{2}
\log\Bigg(
1+\frac{(1-\alpha)P}
{(1+\rho)(\sigma^2+\alpha P)}
\Bigg).
\label{eq:Psi_Leg_cc_G}
\end{align}
The corresponding memoryless mismatched expurgated exponent is
\begin{align}
E_{{\rm mm},{\rm ex}}^{\rm G}(R)
=
\sup_{\rho\ge 1}
\left\{
\Psi_{\rm ex}(\rho)-\rho R
\right\}.
\end{align}
For the i.i.d. ensemble,
\begin{align}\label{eq:Psi_Leg_ex_iid_G}
\Psi_{\rm ex}^{\rm iid}(\rho)
=
\sup_{\substack{s\ge 0:\\ \Delta_{\rm ex}(s,\rho)>0}}
\frac{\rho}{2}
\log
\left(
\frac{\Delta_{\rm ex}(s,\rho)}
{\rho^2\sigma^4}
\right),
\qquad \rho\ge 1,
\end{align}
where
\begin{align}
\Delta_{\rm ex}(s,\rho)
&=
\rho^2\sigma^4
+
2\rho P\sigma^2\sqrt{1-\alpha} s
\twocolbreak \includeonetwocol{}{\quad}
-
s^2
\left[
2\rho P(\sigma^2+\alpha P)
+
P^2(1-\sqrt{1-\alpha})^2
\right].
\end{align}
For the cost-constrained ensemble,
\begin{align}
\Psi_{\rm ex}^{\rm cc}(\rho)
=
\frac{\rho}{2}
\log\Bigg(
1+\frac{(1-\alpha)P}
{2\rho(\sigma^2+\alpha P)}
\Bigg).
\label{eq:Psi_Leg_ex_cc_G}
\end{align}

For the actual fixed-codebook induced channel, 
the operational lower bounds satisfy
\begin{align}
E_{\rm op}^{\rm G}(R)
&\ge
\min\{E_{\rm mm}^{\rm G}(R),E_{\rm res}\},\\
E_{{\rm op},{\rm ex}}^{\rm G}(R)
&\ge
\min\{E_{{\rm mm},{\rm ex}}^{\rm G}(R),E_{\rm res}\}.
\end{align}
\end{lemma}
\begin{proof}
The i.i.d. random-coding and expurgated expressions are obtained by
substituting the ensemble-averaged induced Gaussian channel and the mismatched AWGN decoding
metric into the continuous-alphabet versions of
\eqref{mismatched_exponent_iid} and
\eqref{mismatched_exponent_iid_ex}, respectively. The positivity conditions
$\Delta(s,\rho)>0$ and $\Delta_{\rm ex}(s,\rho)>0$ ensure that the resulting
Gaussian integrals are finite.

For the cost-constrained ensemble, all codewords have identical energy. With
a quadratic auxiliary cost function $a(x)$, the mismatched Gaussian decoding
metric becomes equivalent, up to terms independent of the candidate codeword,
to the metric 
corresponding to the ensemble-averaged induced Gaussian channel. Hence, the
cost-constrained random-coding and expurgated exponents reduce to those of
the matched Gaussian channel with signal power $(1-\alpha)P$ and
noise variance $\sigma^2+\alpha P$.

The operational bounds follow from the total-variation transfer argument used in Proposition~\ref{proposition_exponent_transfer}.
\end{proof}

\begin{lemma}[Gaussian Trojan Detection error exponents]\label{lemma_Gaussian_detection_exponents_codebook}
For Trojan detection in a Gaussian Trojan-infested system, assume a typical
legitimate codeword $x^n$ satisfying
\begin{align}
\frac{1}{n}\sum_{i=1}^n x_i^2 \to P.
\end{align}
When the 
ensemble-averaged induced Gaussian channel is known, the Stein missed-detection exponent is
\begin{align}
E_{\rm M}^{\rm Stein}
=
\frac12
\left[
\log\frac{\sigma^2}{\sigma^2+\alpha P}
+
\frac{2(1-\sqrt{1-\alpha})P}{\sigma^2}
\right].
\end{align}

For the Bayesian formulation, the Chernoff detection exponent is
\begin{align}
E_{\rm e}^{\rm Ch}
=
\max_{0\le \lambda\le 1}
&\Bigg\{
\frac12
\log
\frac{\sigma^2+\lambda\alpha P
}{
(\sigma^2)^{1-\lambda}(\sigma^2+\alpha P)^\lambda
}
\twocolbreak \includeonetwocol{}{\quad}
+
\frac{\lambda(1-\lambda)(1-\sqrt{1-\alpha})^2P}
{2\big(\lambda (\sigma^2+\alpha P)+(1-\lambda)\sigma^2\big)}
\Bigg\}.
\end{align}

Moreover, for a threshold sequence $\eta_n=n\tau$, the false-alarm and
missed-detection exponents of the likelihood-ratio test are
\begin{equation}
\begin{aligned}
E_{\rm FA}(\tau)
&=
\sup_{\theta\ge 0}
\left\{
\theta\tau-\Lambda_{\rm FA}(\theta)
\right\},\\
E_{\rm M}(\tau)
&=
\sup_{\theta\le 0}
\left\{
\theta\tau-\Lambda_{\rm M}(\theta)
\right\},
\end{aligned}
\end{equation}
where
\begin{equation}
\begin{aligned}
\Lambda_{\rm FA}(\theta)
&=
\lim_{n\to\infty}
\frac1n
\log
\mathbb E_{H_0}
\left[
e^{\theta S_n}\mid x^n
\right],\\
\Lambda_{\rm M}(\theta)
&=
\lim_{n\to\infty}
\frac1n
\log
\mathbb E_{H_1}
\left[
e^{\theta S_n}\mid x^n
\right].
\end{aligned}
\end{equation}

In the nominal-channel-only knowledge
case, 
consider the residual-energy
goodness-of-fit detector
\begin{align}
\widehat{\sigma}^2
=
\frac{1}{n}
\sum_{i=1}^n
(Y'_i-x_i)^2 ,
\end{align}
which accepts $H_0$ when
\begin{align}
\widehat{\sigma}^2\in\mathcal A,
\qquad
\mathcal A=[\sigma^2-\tau,\sigma^2+\tau].
\end{align}
For this detector, the false-alarm and missed-detection exponents are
\begin{equation}
\begin{aligned}
E_{\rm FA}
&=
\inf_{r\notin \mathcal A}
\frac12
\left[
\frac{r}{\sigma^2}
-1
-\log\frac{r}{\sigma^2}
\right],
\\
E_{\rm M}
&=
\inf_{r\in \mathcal A}
\sup_{\theta<\frac{1}{2(\sigma^2+\alpha P)}}
\big\{
\theta r
+
\frac12\log(1-2\theta(\sigma^2+\alpha P))
\includeonetwocol{}{\\&\quad}
-
\frac{\theta(\sqrt{1-\alpha}-1)^2P}
{1-2\theta(\sigma^2+\alpha P)}
\big\}.\label{Eq:Sanov_exponents_G}
\end{aligned}
\end{equation}
\end{lemma}

\begin{proof}
Under induced-channel knowledge,
the two hypotheses are Gaussian with
conditional laws
\begin{equation}
\begin{aligned}
H_0 &: Y'_i\sim \mathcal N(x_i,\sigma^2),\\
H_1 &: Y'_i\sim \mathcal N(\sqrt{1-\alpha} \,x_i,\sigma^2+\alpha P).
\end{aligned}
\end{equation}
Thus, conditioned on $x^n$, the Stein exponent is the normalized divergence
from $H_1$ to $H_0$. Using
$\frac1n\sum_i x_i^2\to P$ gives the stated expression. The Chernoff exponent
is obtained by evaluating the Chernoff information between the same Gaussian
product measures. The threshold exponents follow from the large-deviation
formulas in~\eqref{Exponent_Large_Deviation}.

Under nominal-channel-only knowledge,
the detector uses only 
the residual
energy relative to the nominal Gaussian channel. Under $H_0$,
$Y'_i-x_i\sim \mathcal N(0,\sigma^2)$, so Cramer's theorem for the empirical
second moment gives the false-alarm exponent \cite{dembo2009large}. Under $H_1$,
\begin{align}
Y'_i-x_i=(\sqrt{1-\alpha}-1)x_i+Z'_i,
\end{align}
where $Z'_i\sim\mathcal N(0,\sigma^2+\alpha P)$. The logarithmic moment generating
function of the residual energy converges to
\begin{align}
-\frac12\log(1-2\theta(\sigma^2+\alpha P))
+
\frac{\theta(\sqrt{1-\alpha}-1)^2P}{1-2\theta(\sigma^2+\alpha P)},
\end{align}
for any $\theta<\frac{1}{2(\sigma^2+\alpha P)}$. The stated missed-detection exponent follows by the Legendre transform and
minimization over the acceptance region $\mathcal A$. These exponents are
specific to the residual-energy goodness-of-fit detector.
\end{proof}

\section{Numerical Results}
\label{Sec:Simulations}
In this section, we provide numerical results for achievable rates and Trojan detection in Trojan-infested BSC and Gaussian channel models; the numerical results for the legitimate receiver correspond to the ensemble-averaged channel.

\begin{figure*}[t]
\centering
~~\begin{minipage}[t]{0.33\textwidth}
\centering
\includegraphics[trim={2.4cm 0 0 0},scale=0.44]{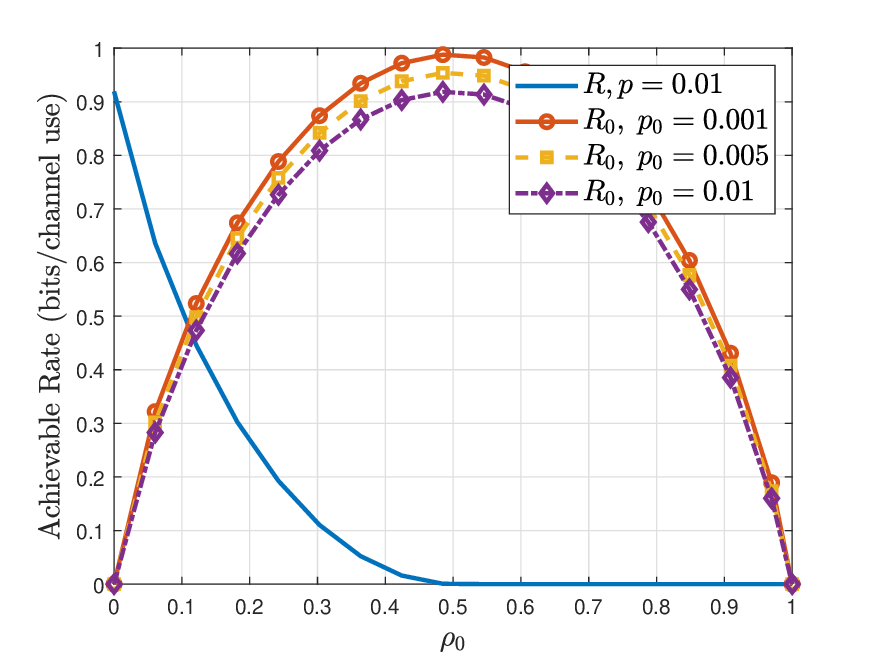}
\caption{Legitimate and rogue achievable rates versus $\rho_0$ in a BSC system.}
\label{fig:BSC2}
\end{minipage}~
\begin{minipage}[t]{0.33\textwidth}
\centering
\includegraphics[trim={1.2cm 0cm 0 1.2cm},scale=0.44]{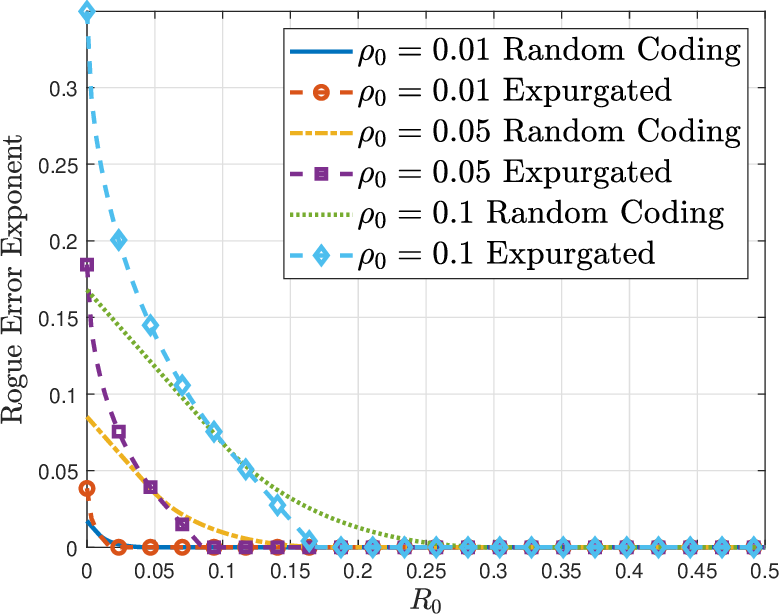}
\caption{Random-coding and expurgated error exponents at the rogue receiver in a BSC system with $p_0=0.005$.
}
\label{fig:BSC_rogue_exp}
\end{minipage}~
\begin{minipage}[t]{0.33\textwidth}
\centering
\includegraphics[trim={0.2cm 0 0 0},scale=0.44]{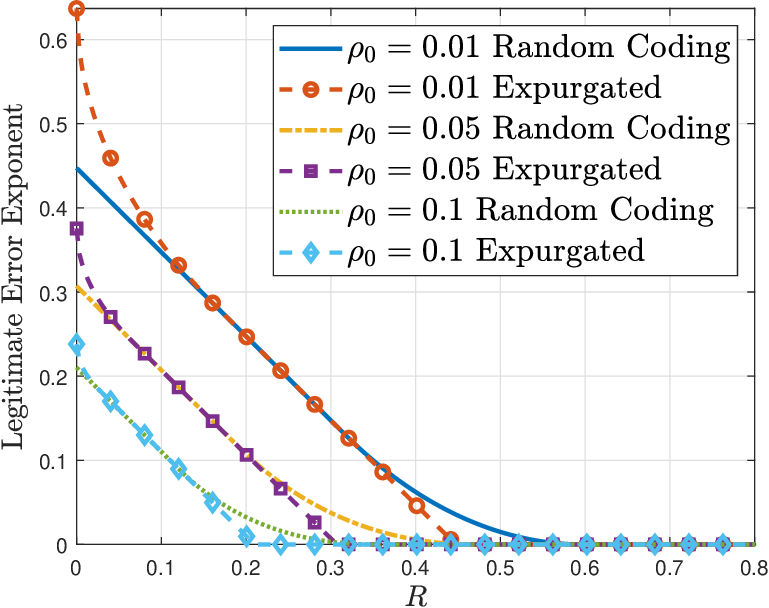}
\caption{Random-coding and expurgated error exponents at the legitimate receiver for a BSC system with $p=0.01$.} 
\label{fig:BSC_legitimate_exp}
\end{minipage}
\vspace{-0.5em}
\end{figure*}

\begin{figure*}[t]
\centering
~~\begin{minipage}[t]{0.33\textwidth}
\centering
\includegraphics[trim={2.4cm 0cm 0 1.2cm},scale=0.44]{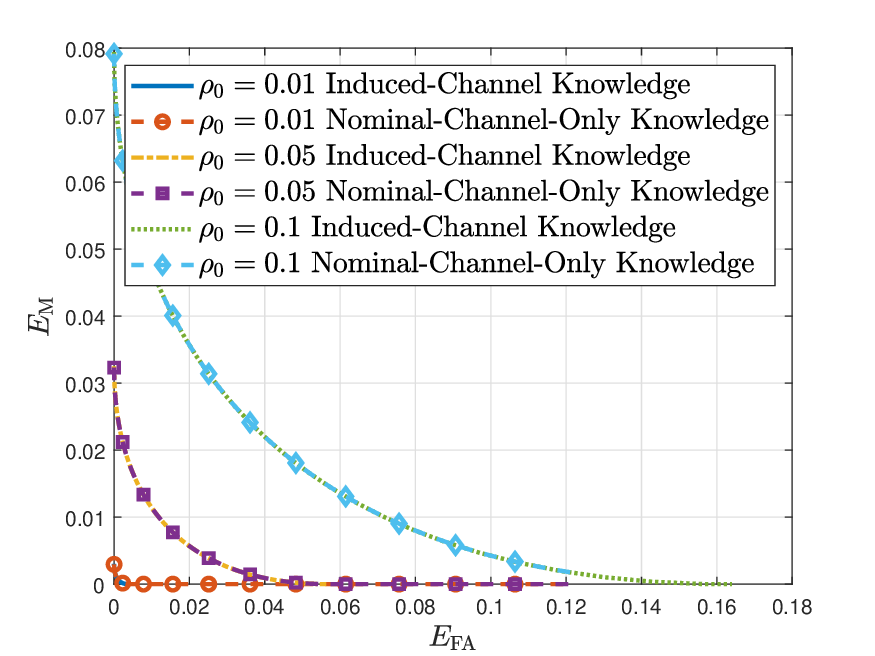}
\caption{$E_{\rm M}$ versus $E_{\rm FA}$ for Trojan detection at the legitimate receiver in a BSC system with $p=0.01$. 
}
\label{fig:BSC_detection_tradeoff}
\end{minipage}~
\begin{minipage}[t]{0.33\textwidth}
\centering
\includegraphics[trim={1.6cm 0 0 0},scale=0.44]{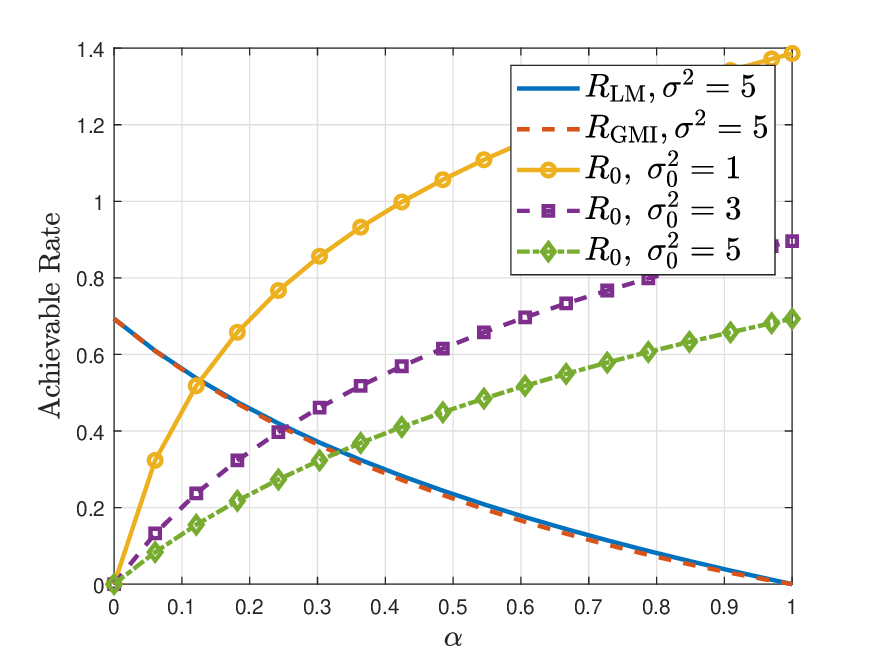}
\caption{Legitimate and rogue achievable rates versus $\alpha$ in a Gaussian channel.}
\label{fig:Gaussian2}
\end{minipage}~
\begin{minipage}[t]{0.33\textwidth}
\centering
\includegraphics[trim={0.8cm 0 0 0}, scale=0.44]{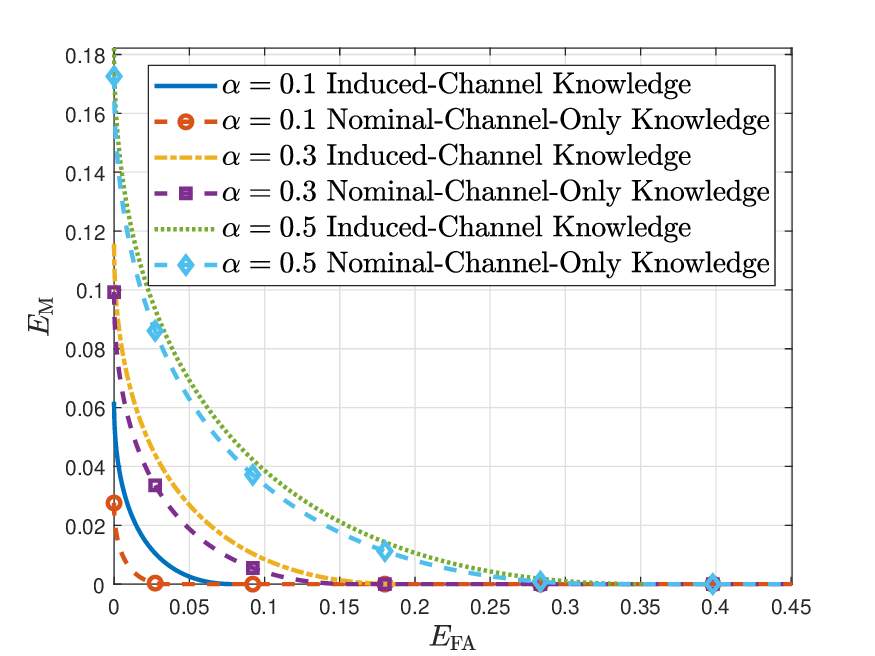}
\caption{$E_{\rm M}$ versus $E_{\rm FA}$ for Trojan detection at the legitimate receiver in a Gaussian channel with $\sigma^2=5$.} 
\label{fig:Gaussian_detection_tradeoff}
\end{minipage}
\vspace{-1em}
\end{figure*}

\subsection{Binary Symmetric Channel}
Here, we set parameters $p=0.01$ and $p_0 = 0.001, 0.005, 0.01$. Since the rogue message rate is zero for $p_0>p$, we restrict the values of $p_0$ to those satisfying $p_0\leq p$. 
Fig.~\ref{fig:BSC2} shows the trade-off between legitimate message and rogue message rates as $\rho_0$ varies for different values of $p_0$. As expected, smaller values of $p_0$ increase the rogue achievable rate. Moreover, increasing $\rho_0$ increases the rogue message rate but reduces the legitimate message rate due to the increased mismatch introduced by the Trojan signaling. In particular, the legitimate message rate becomes zero for $\rho_0>\frac{1}{2}$ according to \eqref{eq:BSC_Rates}, while the rogue message rate is maximized at $\rho_0=\frac{1}{2}$. 

Fig.~\ref{fig:BSC_rogue_exp} shows the achievable reliability exponents at the rogue receiver versus the rogue message rate $R_0$ for different values of $\rho_0$.
As expected, both the random-coding and expurgated exponents decrease as $R_0$ increases. For a fixed $R_0$, both exponents increase with $\rho_0$. The expurgated exponent provides an improvement primarily in the low-rate regime, while the random-coding exponent dominates at moderate and high rates.

Fig.~\ref{fig:BSC_legitimate_exp} illustrates the achievable reliability exponents at the legitimate receiver versus the legitimate message rate $R$ for different values of $\rho_0$. 
Similar to the rogue receiver, both exponents decrease as $R$ increases, and the expurgated exponent mainly improves the low-rate regime. However, unlike the rogue receiver, both the random-coding and expurgated exponents decrease with $\rho_0$ due to the increased mismatch introduced by the Trojan signaling.

Fig.~\ref{fig:BSC_detection_tradeoff} illustrates the tradeoff between the false-alarm exponent $E_{\rm FA}$ and the missed-detection exponent $E_{\rm M}$ for Trojan detection at the legitimate receiver. Increasing the false-alarm exponent imposes a stricter constraint on false alarms and therefore reduces the achievable missed-detection exponent. Conversely, allowing a smaller false-alarm exponent increases the achievable decay rate of missed detections. 
This demonstrates the fundamental detection tradeoff governed by the large-deviation characterization of binary hypothesis testing.
The curves corresponding to the 
induced-channel knowledge
and nominal-channel-only knowledge cases
coincide. Although the two detectors are formulated differently, in the considered BSC setting both ultimately reduce to threshold tests on the empirical crossover frequency between the transmitted and received sequences. Therefore, both detectors induce the same large-deviation events and yield the same asymptotic error exponents.

\subsection{Gaussian Channel}
For the Gaussian channel, we set $\sigma^2=5$, $\sigma_0^2=1,3,5$, and $P=15$. 
Fig.~\ref{fig:Gaussian2} shows the achievable rates versus $\alpha$. 
The rogue achievable rate increases as $\sigma_0^2$ decreases.
Moreover, increasing $\alpha$ increases the rogue message rate at the expense of the legitimate message rate. In particular, the maximum rogue message rate is achieved at $\alpha=1$, where the legitimate achievable rate becomes zero.

Fig.~\ref{fig:Gaussian_detection_tradeoff} illustrates the tradeoff between the false-alarm exponent $E_{\rm FA}$ and the missed-detection exponent $E_{\rm M}$ for Trojan detection in Gaussian channels. 
Unlike the BSC, the induced-channel knowledge and nominal-channel-only knowledge curves do not coincide. Knowledge of the induced channel allows the detector to exploit the statistical effect of the Trojan contamination, 
which can yield a larger missed-detection exponent for a given false-alarm exponent.

\section{Conclusion}\label{Sec:Conclusion}

In this paper, we introduced the Transmitter Trojan channel, in which a malicious encoder embedded within a legitimate transmitter perturbs the transmitted codeword on a symbol-by-symbol basis. The model leads to a broadcast-like channel with one-sided decoder mismatch: the rogue receiver is matched to the 
Trojan channel and the Trojan encoder, while the legitimate receiver continues to use the metric designed for the uncontaminated channel. We derived achievable leakage rates for the Trojan and achievable rates for the concurrent legitimate communication link using broadcast coding, mismatched decoding, and channel-resolvability arguments. We also gave a genie-aided outer bound that is tight when the rogue receiver can decode the legitimate message.

We further studied 
reliability exponents and Trojan detection error exponents. For the legitimate receiver, memoryless mismatched exponents were transferred to operational fixed-codebook lower bounds through an exponential resolvability argument. For the rogue receiver, the reliability analysis separated into decodable, resolvable, and intermediate regimes. For detection at the legitimate receiver, we derived likelihood-ratio exponents when the ensemble-averaged induced channel
is known and detector-specific goodness-of-fit exponents when only the nominal channel is known. Specialization to the BSC and AWGN channel models, together with numerical results, illustrates the resulting tradeoffs among 
rogue message rate, legitimate message rate, decoding reliability, and Trojan detectability.

\appendices
\renewcommand{\thesubsection}{\Alph{section}.\arabic{subsection}}
\renewcommand{\thesubsectiondis}{\thesection.\arabic{subsection}}




\section{Mismatched Rate Calculation}
\label{Appendix-MismatchedRates}
We present the GMI and LM achievable rates for discrete memoryless channels and their extensions to channels with memory. These quantities are lower bounds on the mismatched capacity in general, although they are tight in some special cases. For continuous-alphabet channels, the corresponding expressions are obtained by replacing the summations with integrals.


\subsection{GMI Rate}\label{Appendix-GMIMismatchedRates}
For a memoryless (induced) channel $u(y'|x)$ and decoding metric $w(y'|x)$, the GMI achievable rate can be obtained as follows \cite{scarlett2020information}:
\begin{align}
\twocolAlignMarker
I_{\rm GMI}(X;Y')
\twocolbreak \includeonetwocol{&}{\quad}
=\sup_{s\geq 0}
\sum_{x\in\mathcal X}\sum_{y'\in\mathcal Y}
p_X(x)u(y'|x) \log \frac{w(y'|x)^s}
{\sum_{\bar x\in\mathcal X}p_X(\bar x)w(y'|\bar x)^s}
\nonumber\\\label{eq_GMI}
&\includeonetwocol{}{\quad}
= I(X;Y') - \inf_{s\geq 0}
\mathcal D\big(u'(x|y')\|g_s(x|y')\,|\,p_{Y'}(y')\big).
\end{align}
where $I(X;Y')$, also called the {\em matched rate}, is evaluated with respect to the distribution $u(y'|x)p_X(x)$, and $u'(x|y')$ and $g_s(x|y')$ are defined:
\begin{align}\label{u_prime}
&u'(x|y')\coloneqq \frac{p_X(x)u(y'|x)}{\sum_{\bar{x} \in \mathcal{X}}p_X(\bar{x})u(y'|\bar{x})},\\ \label{v_prime}
&g_s(x|y') \coloneqq \frac{p_X(x)w(y'|x)^s}
{\sum_{\bar{x} \in \mathcal{X}}p_X(\bar{x})w(y'|\bar{x})^s}.
\end{align}
$\mathcal{D}(u'(x|y')||g_s(x|y')|p_{Y'}(y'))$ is the conditional KL divergence:
\begin{align}
&\mathcal{D}(u'(x|y')||g_s(x|y')|p_{Y'}(y'))
\twocolbreak \includeonetwocol{}{\quad}
=\sum_{x \in \mathcal{X}}\sum_{y' \in \mathcal{Y}} p_{Y'}(y')u'(x|y')\log \bigg(\frac{u'(x|y')}{g_s(x|y')}\bigg).
\end{align}

Furthermore, 
if the induced channel has memory with transition probability 
${\widecheck u}_{{\mathcal C}_0} (y'^n|x^n)$, the mismatched rate can be extended as $\lim_{n \rightarrow \infty}\frac{1}{n}I_{\rm GMI,\mathcal{C}_0}(X^n;Y'^n)$ \cite{ganti2000mismatched}, where 
\begin{align}
\twocolAlignMarker
I_{{\rm GMI},\mathcal{C}_0}(X^n;Y'^n)
\twocolbreak \includeonetwocol{}{\quad}
=\sup_{s^{(n)}\geq 0}\sum_{\substack{x^n \in \mathcal{X}^n,\\
y'^n \in \mathcal{Y}^n}}\big(\prod_{i=1}^{n}p_X(x_i)\big){\widecheck u}_{{\mathcal C}_0} (y'^n|x^n)\label{eq_GMI_memory}
\twocolbreak \includeonetwocol{}{\qquad \times}\log\bigg(\frac{\big(\prod_{i=1}^{n}w(y'_i|x_i)\big)^{s^{(n)}} }{\sum_{\bar{x}^n \in \mathcal{X}^n}\prod_{i=1}^{n}p_X(\bar{x}_i) w(y'_i|\bar{x}_i)^{s^{(n)}} }\bigg).
\end{align}

\subsection{LM Rate}\label{Appendix-LMMismatchedRates}
For a memoryless (induced) channel $u(y'|x)$ and decoding metric $w(y'|x)$, the LM achievable rate can be obtained as follows \cite{scarlett2020information}:
\begin{align}
\twocolAlignMarker
I_{\rm LM}(X;Y')\twocolbreak \includeonetwocol{&}{}
=
\sup_{\substack{s\geq 0,\\ a(\cdot)}}
\sum_{x\in\mathcal X}\sum_{y'\in\mathcal Y}
p_X(x)u(y'|x)
\log
\frac{w(y'|x)^s e^{a(x)}}
{\sum_{\bar x\in\mathcal X}
p_X(\bar x)w(y'|\bar x)^s e^{a(\bar x)}}
\nonumber\\\label{eq_LM}
&=
I(X;Y')
-
\inf_{s\geq 0,a(\cdot)}
\mathcal D\big(u'(x|y')\|g'_{s,a}(x|y')\,|\,p_{Y'}(y')\big),
\end{align}
with
\begin{align}
g'_{s,a}(x|y')
\coloneqq
\frac{p_X(x)w(y'|x)^s e^{a(x)}}
{\sum_{\bar x\in\mathcal X}
p_X(\bar x)w(y'|\bar x)^s e^{a(\bar x)}}.
\label{v_double_prime}
\end{align}

Similar to the GMI rate, the LM rate can be extended to the 
induced channel with memory as $\lim_{n \rightarrow \infty} \frac{1}{n} I_{\rm LM,\mathcal{C}_0}(X^n;Y'^n)$, where
\begin{align}
\twocolAlignMarker
I_{{\rm LM},\mathcal{C}_0}(X^n;Y'^n)
\twocolbreak
\includeonetwocol{&}{}=\sup_{\substack{s^{(n)}\geq 0, \\a^{(n)}(.)}}\sum_{\substack{x^n \in \mathcal{X}^n,\\
y'^n \in \mathcal{Y}^n}}\big(\prod_{i=1}^{n}p_X(x_i)\big){\widecheck u}_{{\mathcal C}_0} (y'^n|x^n)
\twocolbreak \includeonetwocol{}{\quad \times}
\log\bigg(\frac{\big(\prod_{i=1}^{n}w(y'_i|x_i)\big)^{s^{(n)}}e^{a^{(n)}(x^n)} }{\sum_{\bar{x}^n \in \mathcal{X}^n}\big(\prod_{i=1}^{n}p_X(\bar{x}_i)w(y'_i|\bar{x}_i)^{s^{(n)}}\big)e^{a^{(n)}(\bar{x}^n)} }\bigg).\label{eq_LM_memory}
\end{align}

The divergence representations in~\eqref{eq_GMI} and~\eqref{eq_LM} show that
for the induced channel $u(y'|x)$ and the
fixed decoding metric $w(y'|x)$,
\begin{align}
I_{\rm GMI}(X;Y')\le I(X;Y'),
\qquad
I_{\rm LM}(X;Y')\le I(X;Y').
\end{align}
Indeed, the difference between the matched mutual information $I(X;Y')$ and
the corresponding mismatched achievable rate is the infimum of a conditional
KL divergence, and is therefore nonnegative. This inequality is used in
Appendix~\ref{Appendix_less_noisy_rogue}.

The legitimate rate loss can be interpreted as having two components. The
first is the {\em Trojan channel mismatch penalty}, represented by the
nonnegative divergence terms in~\eqref{eq_GMI} and~\eqref{eq_LM}. These terms
appear because the Trojan changes the effective channel from $X$ to $Y'$,
while the legitimate decoder continues to use the metric matched to the
original channel $w(y'|x)$.

The second component is the {\em Trojan signal-path penalty}. Even if the
decoder were matched to the induced channel $u(y'|x)$, the mutual-information
term in~\eqref{eq_GMI} and~\eqref{eq_LM} would be $I(X;Y')$, computed across
the Trojan-modified signal path
\begin{align}
X \to X' \to Y'.
\end{align}
By data processing,
\begin{align}
I(X;Y')\le I(X';Y').
\end{align}
Thus the Trojan perturbation can reduce the information carried by the
legitimate codeword before any decoder mismatch is taken into account. In
this sense, the loss in the legitimate link consists of a signal-path
degradation caused by embedding the rogue signal, followed by an additional
mismatch penalty caused by the unchanged legitimate decoder.

When the uncontaminated codebook was designed near the optimum of the
original channel, this signal-path degradation is naturally interpreted as a
loss relative to the link margin available in the Trojan-free system.


\section{Fixed-codebook achievable rate}
 \label{Appendix_Marton_Mismatch1}

The 
original legitimate codebook $\mathcal C$ is drawn according to distribution $\prod_{i=1}^np_X^*(x_i)$, with $e^{nR}$ codewords $x^n(M)$, $M \in [1:e^{nR}]$. The Trojan codebook ${\mathcal C}_0$, is generated i.i.d.\ according to $\prod_{i=1}^n p_V(v_i)$ with $e^{nR_0}$ codewords $v^n(M_0)$, $M_0 \in [1:e^{nR_0}]$.
The codewords $x^n(M)$ and $v^n(M_0)$ are combined with a symbol-by-symbol function $h(.,.)$ to produce the 
codeword $X'^n$ emitted from the contaminated transmitter. Because the contaminated transmitter emits into the same channel as the uncontaminated transmitter, $X$ and $X'$ share the same alphabet.

The rogue message achievable rate $R_0 \leq I(V;Y_0)$ follows directly from the Marton inner bound, since the Trojan channel input is generated via a symbol-by-symbol function. No sum-rate expression appears since the legitimate and Trojan codebooks are independently generated.

The legitimate message rate is governed by the effective channel law from $X$ to $Y'$ subject to a fixed Trojan codebook, i.e., \eqref{dist_u_memory}. 
Since the Trojan operation introduces memory into the channel, the legitimate message rate must be expressed in a multi-letter form for channels with memory, as provided in \eqref{eq:MartonRate1-mismatch1}.
By applying the GMI and LM mismatched rate formulations for channels with memory, given in \eqref{eq_GMI_memory} and \eqref{eq_LM_memory}, we obtain $I_{\textrm{M},\mathcal{C}_0}(X^n;Y'^n)$.

\section{Single-Letter Achievable Rates Under Channel Resolvability}
\label{Appendix_Marton_Mismatch2}

We prove Theorem~\ref{thm-singleletter-resolvability} in three steps. First, we define a randomized Trojan codebook. Second, we use a resolvability argument to show that the channel induced by the Trojan from $X^n$ to $Y'^n$ is close in total variation to a product channel with high probability. Third, we use the stability of the GMI and LM mismatched-rate functionals under total-variation approximation to obtain the stated single-letter achievable rate.

\subsection{Randomized Trojan Codebook and Induced Channel}

We follow the methodology of Cuff~\cite{cuff2016soft} in defining and analyzing a random codebook for the Trojan. Let $\mathsf C_0=\{V^n(m_0,k)\}$
denote a random Trojan codebook where the codewords are mutually independent and each has distribution $\prod_{i=1}^n p_V(v_i)$.
The index
$m_0\in[1:e^{nR_0}]$ carries the rogue message and 
$k\in[1:e^{nR'}]$ is an additional randomization index known to the encoder and the rogue receiver. Define
\begin{align}
R_0' \coloneqq R_0+R'.
\end{align}
For a fixed legitimate codeword $x^n$, the randomized Trojan codebook induces
the conditional distribution
\begin{align}
\widecheck u_{\mathsf C_0}(y'^n|x^n)
&=
\frac{1}{e^{nR_0'}}
\sum_{m_0=1}^{e^{nR_0}}
\sum_{k=1}^{e^{nR'}}
\prod_{i=1}^n
w(y'_i|h(x_i,V_i(m_0,k))).
\end{align}
Equivalently, indexing the $e^{nR_0'}$ codewords by a single index $j$, this
can be written as
\begin{align}
\widecheck u_{\mathsf C_0}(y'^n|x^n)
=
\frac{1}{e^{nR_0'}}
\sum_{j=1}^{e^{nR_0'}}
\prod_{i=1}^n
w(y'_i|h(x_i,V_i(j))).
\label{eq:check_u_randomized}
\end{align}
We aim to show that $\widecheck u_{\mathsf C_0}(y'^n|x^n)$ is close in total variation to the memoryless conditional distribution $u^n(y'^n|x^n)\coloneqq \prod_{i=1}^n u(y'_i|x_i)$, where
\begin{align}
u(y'|x)=\sum_v p_V(v)w(y'|h(x,v)).
\label{eq:target_memoryless_u}
\end{align}

\subsection{Conditional Resolvability}
\label{conditional_resolvability}
We now show that the randomized Trojan codebook can approximate the
memoryless conditional distribution in~\eqref{eq:target_memoryless_u}. This
is a one-sided resolvability problem for the virtual channel from $V$ to $Y'$, conditioned on $X=x$.

\begin{proposition}[Conditional resolvability]
\label{prop_conditional_resolvability}
For every fixed $x^n$ whose empirical distribution approaches $p_X^*$, if
\begin{align}
R_0' > I(V;Y'|X),
\end{align}
then there exist constants $\delta_1,\delta_2>0$ such that
\begin{align}
\mathbb P_{\mathsf C_0}
\left(
\left\|
\widecheck u_{\mathcal C_0}(\cdot|x^n)
-
u^n(\cdot|x^n)
\right\|_{\rm TV}
>
e^{-n\delta_1}
\right)
\le
e^{-e^{n\delta_2}}.
\label{eq:conditional_resolvability_tail}
\end{align}
Consequently,
\begin{align}
\left\|
\widecheck u_{\mathsf C_0}(\cdot|x^n)
-
u^n(\cdot|x^n)
\right\|_{\rm TV}
\to 0
\end{align}
almost surely with respect to the random Trojan codebook ensemble.
\end{proposition}

\begin{proof}
Using the positive-part representation of total variation,
\begin{align}
&\left\|
\widecheck u_{\mathsf C_0}(
\cdot|x^n)
-
u^n(\cdot|x^n)
\right\|_{\rm TV}
\twocolbreak
\includeonetwocol{}{\quad}
=
\sum_{y'^n\in\mathcal Y^n}
\left[
\widecheck u_{\mathsf C_0}(y'^n|x^n)
-
u^n(y'^n|x^n)
\right]^+ .
\end{align}
Substituting~\eqref{eq:check_u_randomized} gives
\begin{align}
&\left\|
\widecheck u_{\mathsf C_0}(
\cdot|x^n)
-
u^n(\cdot|x^n)
\right\|_{\rm TV}
\twocolbreak
\includeonetwocol{}{\quad}
=
\sum_{y'^n\in\mathcal Y^n}
\Bigg[
\frac{1}{e^{nR_0'}}
\sum_{j=1}^{e^{nR_0'}}
p_{Y'^n|X^n,V^n}(y'^n|x^n,V^n(j))
\twocolbreak
\includeonetwocol{}{\qquad} -
u^n(y'^n|x^n)
\Bigg]^+ .
\end{align}
Let $\mathcal T_\epsilon^{(n)}$ denote the jointly typical set with respect to
$p_X^*(x)p_V(v)w(y'|h(x,v))$. Splitting the summation into typical and
atypical parts yields
\begin{align}
&\left\|
\widecheck u_{\mathsf C_0}(\cdot|x^n)
-
u^n(\cdot|x^n)
\right\|_{\rm TV}
\twocolbreak
\includeonetwocol{}{\quad}
\le
P_{\rm atyp}
+
\sum_{y'^n\in\mathcal Y^n}
u^n(y'^n|x^n)
\left[P_{\rm typ}(y'^n)-1\right]^+,
\label{eq:TV_typ_atyp_split}
\end{align}
where
\begin{align}
P_{\rm atyp}
&=
\frac{1}{e^{nR_0'}}
\sum_{j=1}^{e^{nR_0'}}
\sum_{y'^n\in\mathcal Y^n}
p_{Y'^n|X^n,V^n}(y'^n|x^n,V^n(j))
\twocolbreak
\includeonetwocol{}{\qquad\qquad\qquad\qquad
\times}
\, \mathbbm{1}_{\{(x^n,V^n(j),y'^n)\notin\mathcal T_\epsilon^{(n)}\}},
\end{align}
and
\begin{align}
P_{\rm typ}(y'^n)
&=
\frac{1}{e^{nR_0'}}
\sum_{j=1}^{e^{nR_0'}}
\frac{
p_{Y'^n|X^n,V^n}(y'^n|x^n,V^n(j))
}{
u^n(y'^n|x^n)
}
\twocolbreak
\includeonetwocol{}{\qquad\qquad\quad
\times}
\, \mathbbm{1}_{\{(x^n,V^n(j),y'^n)\in\mathcal T_\epsilon^{(n)}\}} .
\end{align}
The inequality in~\eqref{eq:TV_typ_atyp_split} follows from
$[a+b-c]^+\le a+[b-c]^+$ for $a\ge 0$.

For any $\delta_1>0$, the union bound gives
\begin{align}
\twocolAlignMarker
\mathbb P_{\mathsf C_0}
\left(
\left\|
\widecheck u_{\mathsf C_0}(\cdot|x^n)
-
u^n(\cdot|x^n)
\right\|_{\rm TV}
>
e^{-n\delta_1}
\right)
\twocolbreak
\includeonetwocol{}{\quad}
\le
\mathbb P_{\mathsf C_0}
\left(
P_{\rm atyp}>\frac12 e^{-n\delta_1}
\right)
\twocolbreak
\includeonetwocol{}{\qquad}
+
\sum_{y'^n\in\mathcal Y^n}
\mathbb P_{\mathsf C_0}
\left(
P_{\rm typ}(y'^n)>1+\frac12 e^{-n\delta_1}
\right).
\label{eq:union_resolvability}
\end{align}

We first bound the atypical term. Write
\begin{align}
P_{\rm atyp}
=
\frac{1}{e^{nR_0'}}
\sum_{j=1}^{e^{nR_0'}} P_{{\rm atyp},j},
\end{align}
where
\begin{align}
P_{{\rm atyp},j}
&=
\sum_{y'^n\in\mathcal Y^n}
p_{Y'^n|X^n,V^n}(y'^n|x^n,V^n(j))
\twocolbreak
\includeonetwocol{}{\qquad\qquad\qquad
\times}
\, \mathbbm{1}_{\{(x^n,V^n(j),y'^n)\notin\mathcal T_\epsilon^{(n)}\}} .
\end{align}
The random variables $P_{{\rm atyp},j}$ are independent and take values in $[0,1]$. By
standard typicality, for some $\delta_1'>0$,
\begin{align}
\mathbb E[P_{\rm atyp}]
\le
e^{-n\delta_1'} .
\end{align}
Applying Hoeffding's inequality~\cite{hoeffding1963probability}, for
$\delta_1<\delta_1'$ and sufficiently large $n$,
\begin{align}
\mathbb P_{\mathsf C_0}
\left(
P_{\rm atyp}>\frac12 e^{-n\delta_1}
\right)
\le
\exp\left(
-\frac12 e^{n(R_0'-2\delta_1)}
\right).
\label{eq:Patyp_bound}
\end{align}

For the typical term, the standard resolvability concentration bound for the
conditioned channel gives~\cite[Lemma~1]{frey2017mac2}
\begin{align}
&\mathbb P_{\mathsf C_0}
\left(
P_{\rm typ}(y'^n)>1+\frac12 e^{-n\delta_1}
\right)
\twocolbreak
\includeonetwocol{}{\quad}
\le
\exp\left(
-\frac{1}{12}
e^{n(R_0'-I(V;Y'|X)-2\delta_1-\epsilon)}
\right).
\label{eq:Ptyp_bound}
\end{align}
Combining~\eqref{eq:union_resolvability}, \eqref{eq:Patyp_bound}, and
\eqref{eq:Ptyp_bound}, we obtain
\begin{align}
&\mathbb P_{\mathsf C_0}
\left(
\left\|
\widecheck u_{\mathsf C_0}(\cdot|x^n)
-
u^n(\cdot|x^n)
\right\|_{\rm TV}
>
e^{-n\delta_1}
\right)
\twocolbreak
\includeonetwocol{}{\quad}
\le
\exp\left(
-\frac12 e^{n(R_0'-2\delta_1)}
\right)
\twocolbreak
\includeonetwocol{}{\qquad}
+
|\mathcal Y|^n
\exp\left(
-\frac{1}{12}
e^{n(R_0'-I(V;Y'|X)-2\delta_1-\epsilon)}
\right).
\label{eq:double_exp_resolvability}
\end{align}
If $R_0'>I(V;Y'|X)$, then $\delta_1$ and $\epsilon$ can be chosen small enough
so that the right hand side of~\eqref{eq:double_exp_resolvability} is
summable in $n$. The Borel--Cantelli lemma then gives almost sure convergence
of the total-variation distance to zero.
\end{proof}

Since the failure probability in
\eqref{eq:conditional_resolvability_tail} is double exponential in $n$, a
union bound over the exponentially many legitimate codewords shows that the
same exponential total-variation bound holds simultaneously for all
legitimate codewords of typical type with probability approaching one over
the Trojan codebook ensemble. 

\subsection{Stability of the Mismatched Rate Functional}

We next connect the total-variation approximation in
Proposition~\ref{prop_conditional_resolvability} to the single-letter GMI and
LM rates.

\begin{lemma}[Stability of mismatched rates]
\label{lemma_mismatch_stability_TV}
Assume finite alphabets and a strictly positive decoding metric
$w(y'|x)$. If
\begin{align}
\max_{x^n\in\mathcal T_n(p_X^*)}
\left\|
\widecheck u_{\mathcal C_0}(\cdot|x^n)
-
u^n(\cdot|x^n)
\right\|_{\rm TV}
\to 0,
\end{align}
then, for every finite upper bound $s_{\max}>0$ on the LM/GMI optimization
parameter $s$ and every finite truncation level $A>0$ on the auxiliary
function $a(\cdot)$, the normalized truncated $n$-letter GMI and LM
functionals converge to the corresponding single-letter truncated GMI and LM
rates under $u(y'|x)$. Letting $s_{\max}$ and $A$ grow gives the ordinary GMI
and LM achievable rates up to an arbitrary backoff.
\end{lemma}



\begin{proof}
For fixed $s_{\max}<\infty$ and $A<\infty$, the logarithmic decoding metric
appearing in the GMI and LM functionals is uniformly bounded in magnitude by
a constant times $n$, because the alphabets are finite and the decoding metric
is strictly positive. Hence, after normalization by $1/n$, changing the
conditional distribution by $\delta_n$ in total variation changes the
corresponding normalized functional by at most $O(\delta_n)$.

For the product channel $\prod_i u(y'_i|x_i)$ and input sequences whose
empirical distribution approaches $p_X^*$, the normalized $n$-letter
functional factorizes into the corresponding single-letter functional, up to
a vanishing type approximation term. Taking the supremum over the compact
truncated parameter set preserves convergence. Finally, the full GMI and LM
rates are obtained by choosing the truncation constants large enough to
approximate the corresponding suprema within an arbitrary positive backoff.
\end{proof}

\subsection{Completion of the Proof of Theorem~\ref{thm-singleletter-resolvability}}

By Proposition~\ref{prop_conditional_resolvability}, if
\begin{align}
R_0+R' > I(V;Y'|X),
\end{align}
then there exists a sequence of Trojan codebooks  for which the fixed-codebook induced channel $\widecheck u_{\mathcal C_0}(y'^n|x^n)$ is asymptotically
close in total variation to the product channel $\prod_i u(y'_i|x_i)$ for
all legitimate codewords of typical type.

By Lemma~\ref{lemma_mismatch_stability_TV}, the corresponding GMI and LM
achievable rates under the fixed-codebook induced channel are arbitrarily
close to the single-letter GMI and LM rates evaluated under $u(y'|x)$. Hence,
for every $\epsilon>0$, the legitimate receiver supports every rate
\begin{align}
R<I_{\rm M}(X;Y')-\epsilon.
\end{align}

The rogue receiver decodes the rogue message through the channel induced by
the random-coding distribution. For the independent-decoding construction,
the rogue receiver supports every rate
\begin{align}
R_0<I(V;Y_0).
\end{align}
This proves Theorem~\ref{thm-singleletter-resolvability}.

\section{Less-Noisy Rogue Receiver}
 \label{Appendix_less_noisy_rogue}

We prove Proposition~\ref{Proposition_superposition}. Under the randomization condition
of Theorem~\ref{thm-singleletter-resolvability}, the legitimate receiver may be analyzed using
the  
ensemble-averaged induced channel $u(y'|x)$ and the mismatched decoding metric
corresponding to the original channel $w(y'|x)$. Hence the legitimate receiver supports every rate
below $I_{\rm M}(X;Y')$, with the usual arbitrarily small backoff.

It remains to characterize the rogue receiver. Since the rogue receiver is
less noisy than the legitimate receiver, it decodes the legitimate and rogue
messages jointly. Assume without loss of generality that $(M,M_0)=(1,1)$ was
sent. The relevant error events are
\begin{align}
E_1
&=
\{(X^n(1),V^n(1),Y_0^n)\notin \mathcal T_\epsilon^{(n)}\},
\nonumber\\
E_2
&=
\{(X^n(1),V^n(m_0),Y_0^n)\in \mathcal T_\epsilon^{(n)}
\text{ for some }m_0\ne 1\},
\nonumber\\
E_3
&=
\{(X^n(m),V^n(m_0),Y_0^n)\in \mathcal T_\epsilon^{(n)}
\twocolbreak
\includeonetwocol{}{\quad}
\text{ for some }m\ne 1,\ m_0\ne 1\}.
\end{align}
By the law of large numbers, $\mathbb P(E_1)\to 0$. Standard packing
bounds give $\mathbb P(E_2)\to 0$ if
\begin{align}
R_0<I(V;Y_0|X),
\end{align}
and $\mathbb P(E_3)\to 0$ if
\begin{align}
R+R_0<I(X,V;Y_0).
\end{align}
Thus the joint decoder supports the region
\begin{align}
R &< I_{\rm M}(X;Y'),\\
R_0 &< I(V;Y_0|X),\\
R+R_0 &< I(X,V;Y_0),
\label{eq:less_noisy_three_constraint_region}
\end{align}
where the first inequality is the legitimate-rate condition supplied by
Theorem~\ref{thm-singleletter-resolvability}.

We next show that the third inequality is implied by the first two. For the
GMI and LM mismatched achievable rates,
\begin{align}
I_{\rm M}(X;Y')\le I(X;Y'),
\end{align}
because the mismatch-rate expressions can be written as the matched mutual
information under $u(y'|x)$ minus a nonnegative divergence term. Therefore,
using the first two constraints,
\begin{align}
R+R_0
&<
I_{\rm M}(X;Y')+I(V;Y_0|X)
\nonumber\\
&\le
I(X;Y')+I(V;Y_0|X).
\label{eq:less_noisy_step_a}
\end{align}

It remains to compare $I(X;Y')$ and $I(X;Y_0)$. Under the joint distribution
\begin{align}
p_X^*(x)p_V(v)\mathbbm{1}_{\{x'=h(x,v)\}}\widehat w(y',y_0|x'),
\end{align}
the marginal law of $(X,X',Y',Y_0)$ factors as
\begin{align}
p_X^*(x)p(x'|x)\widehat w(y',y_0|x'),
\end{align}
where
\begin{align}
p(x'|x)=\sum_v p_V(v)\mathbbm{1}_{\{x'=h(x,v)\}}.
\end{align}
Hence
\begin{align}
X\to X'\to (Y',Y_0)
\end{align}
forms a Markov chain. Thus $X$ is an admissible auxiliary variable in the
less-noisy condition. Since $Y_0$ is less noisy than $Y'$, we have
\begin{align}
I(X;Y')\le I(X;Y_0).
\label{eq:less_noisy_step_b}
\end{align}
Combining~\eqref{eq:less_noisy_step_a} and~\eqref{eq:less_noisy_step_b}
gives
\begin{align}
R+R_0
&<
I(X;Y_0)+I(V;Y_0|X)
\nonumber\\
&=
I(X,V;Y_0),
\end{align}
so the sum-rate constraint in
\eqref{eq:less_noisy_three_constraint_region} is redundant. This proves the proposition.

\section{Genie-Aided Outer Bound}
\label{Appendix_genie_aided_outer_bound}

Let $P_{\rm e}^{(n)}$ and $P_{{\rm e},0}^{(n)}$ denote the legitimate and rogue decoding
error probabilities, respectively, both of which tend to zero as $n \to \infty$. By Fano's inequality, there exist sequences
$\epsilon_{n}\to0$ and $\epsilon_{0,n}\to0$ such that
\begin{align*}
H(M|Y'^n)&\le n\epsilon_n,\\
H(M_0|Y_0^n)&\le n\epsilon_{0,n}.
\end{align*}

For the legitimate message rate,
\begin{align}\nonumber
nR
&=H(M)\\\nonumber
&\le I(M;Y'^n)+n\epsilon_n\\\nonumber
&\le I(X^n;Y'^n)+n\epsilon_n\\\nonumber
&\le I(X'^n;Y'^n)+n\epsilon_n\\
&\le \sum_{i=1}^n I(X'_i;Y'_i)+n\epsilon_n.
\end{align}
The third inequality follows from data processing, and the last inequality follows from the memoryless channel law $w(y'|x')$.

For the rogue message rate, we give the rogue receiver the legitimate codeword
$X^n$ as genie side information. Then
\begin{align}\nonumber
nR_0
&=H(M_0)\\\nonumber
&\le I(M_0;Y_0^n,X^n)+n\epsilon_{0,n}\\\nonumber
&= I(M_0;Y_0^n|X^n)+n\epsilon_{0,n}\\\nonumber
&\le I(V^n;Y_0^n|X^n)+n\epsilon_{0,n}\\
&\le \sum_{i=1}^n I(V_i;Y_{0i}|X_i)+n\epsilon_{0,n}.
\end{align}
The equality uses the independence of $M_0$ and $X^n$, and the inequalities
use the data processing relation from $M_0$ to $V^n$ and then to $Y_0^n$
given $X^n$, followed by the memoryless channel law.

For the sum-rate bound, reveal neither message to the receivers and use the
pair of observations:
\begin{align}\nonumber
n(R+R_0)
&=H(M,M_0)\\\nonumber
&\le I(M,M_0;Y'^n,Y_0^n)+n\epsilon'_n\\\nonumber
&\le I(X^n,V^n;Y'^n,Y_0^n)+n\epsilon'_n\\
&\le \sum_{i=1}^n I(X_i,V_i;Y'_i,Y_{0i})+n\epsilon'_n,
\end{align}
where $\epsilon'_n\to0$.
Therefore, we have shown that
\begin{align*}\nonumber
R
&\le \frac{1}{n}\sum_{i=1}^n I(X'_i;Y'_i)+\epsilon_n,\\\nonumber
R_0&\le \frac{1}{n}\sum_{i=1}^n I(V_i;Y_{0i}|X_i)+\epsilon_{0,n},\\
R+R_0&\le \frac{1}{n}\sum_{i=1}^n I(X_i,V_i;Y'_i,Y_{0i})+\epsilon'_n,
\end{align*}

Let $Q$ be uniform on $\{1,\ldots,n\}$ and independent of all other random
variables, and define
\begin{align*}
\twocolAlignMarker
X\coloneqq X_Q,\qquad V\coloneqq V_Q,\qquad X'=X'_Q,\twocolbreak
\includeonetwocol{\qquad}{} Y'\coloneqq Y'_Q,\qquad Y_0\coloneqq Y_{0Q}.
\end{align*}

Then
\begin{align*}
R &\le I(X';Y'|Q)+\epsilon_n,\\
R_0 &\le I(V;Y_0|X,Q)+\epsilon_{0,n},\\
R+R_0 &\le I(X,V;Y',Y_0|Q)+\epsilon'_n,
\end{align*}
Since $\epsilon_n, \epsilon_{0,n}, \epsilon'_n$ tend to zero as $n \to \infty$.
Therefore, every achievable rate pair satisfies
\begin{align*}\nonumber
R &\le I(X';Y'|Q),\\\nonumber
R_0 &\le I(V;Y_0|X,Q),\\
R+R_0 &\le I(X,V;Y',Y_0|Q),
\end{align*}
for some distribution of the form
$p(q)p(x,v|q)\mathbbm{1}_{\{x'=h(x,v)\}}\widehat w(y',y_0|x')$,
with $Q$ taking values in some finite set (independent of $n$).

\section{Equivalence of GMI and LM Rates for the BSC}
\label{Appendix_GMI_BSC}
We show that for a BSC $u(y'|x)$, with decoding metric corresponding to another BSC $w(y'|x)$, the GMI rate is equal to the LM rate as in \eqref{LM_binary1}. Using \eqref{eq_GMI} we have:
\begin{align}
\twocolAlignMarker
I_{\rm GMI}^{\rm BSC}(X;Y')
\twocolbreak
\includeonetwocol{&}{}
=\max_{s\geq 0}\bigg[(1-\rho) u(0|0)\log\bigg(\frac{w(0|0)^s}{(1-\rho)w(0|0)^s+\rho w(0|1)^s}\bigg)
\nonumber\\
&\quad
+(1-\rho) u(1|0)\log\bigg(\frac{w(1|0)^s}{(1-\rho) w(1|0)^s+\rho w(1|1)^s }\bigg)\nonumber\\
&\quad
+\rho u(0|1)\log\bigg(\frac{w(0|1)^s}{(1-\rho) w(0|0)^s+\rho w(0|1)^s}\bigg)
\nonumber\\
&\quad+\rho u(1|1)\log\bigg(\frac{w(1|1)^s}{(1-\rho) w(1|0)^s+\rho w(1|1)^s}\bigg)\bigg]\nonumber\\
&=\max_{s\geq 0}\bigg[(1-\rho) u(0|0)\log\bigg(\frac{w(0|0)^s}{(1-\rho)w(0|0)^s+\rho w(0|1)^s}\bigg)
\nonumber\\
&\quad+(1-\rho) u(0|1)\log\bigg(\frac{w(0|1)^s}{(1-\rho) w(0|1)^s+\rho w(0|0)^s }\bigg)\nonumber\\
&\quad +\rho u(0|1)\log\bigg(\frac{w(0|1)^s}{(1-\rho) w(0|0)^s+\rho w(0|1)^s}\bigg)
\nonumber\\
&\quad+\rho u(0|0)\log\bigg(\frac{w(0|0)^s}{(1-\rho) w(0|1)^s+\rho w(0|0)^s}\bigg)\bigg].
\end{align}
By differentiating with respect to $s$ and setting the derivative equal to zero, we obtain 
${s=s_1=\frac{\log\big(\frac{u(0|1)}{u(0|0)}\big)}{\log\big(\frac{w(0|1)}{w(0|0)}\big)}}$.
If $s_1>0$, which occurs if $\textrm{sign}\big(\log\big(\frac{w(0|0)}{w(0|1)}\big)\big)=\textrm{sign}\big(\log\big(\frac{u(0|0)}{u(0|1)}\big)\big)$, or equivalently $\textrm{det}(\mathbf{W}_{Y'|X'})\textrm{det}(\mathbf{U}_{Y'|X})\geq 0$,
then $s^*=s_1$, and the mismatched rate equals the matched rate. Otherwise, the optimum s is $s^*=0$, and the mismatched rate equals zero. This result is consistent with the LM rate in \eqref{LM_binary1}.

\section{BSC Legitimate Error Exponents}
\label{Appendix_BSC_error_exponent}

For the legitimate receiver, 
the i.i.d. random coding exponent simplifies to
\begin{align}
\Psi^{\rm iid}(\rho)
=
\sup_{s\ge 0}
-\log
&\Big[
\left(
\tfrac{(1-p)^s+p^s}{2}
\right)^\rho
\Big(
(1-\widetilde p)(1-p)^{-s\rho}
\twocolbreak
\includeonetwocol{}{\quad}+
\widetilde p p^{-s\rho}
\Big)
\Big].
\end{align}
When $\rho_0\le \frac12$, the effective BSC has the same likelihood ordering
as the original BSC used in the mismatched decoding metric, yielding a positive achievable mismatched rate. Differentiating the exponent expression
with respect to $s$ shows that the optimizing value is
\begin{align}
s^\star
=
\frac{1}{1+\rho}
\frac{\log\frac{1-\widetilde p}{\widetilde p}}
{\log\frac{1-p}{p}} .
\end{align}
Using this value gives
\begin{align}
\Psi^{\rm iid}(\rho)
=
\rho \log 2
-
(1+\rho)
\log
\left[
(1-\widetilde p)^{\frac{1}{1+\rho}}
+
\widetilde p^{\frac{1}{1+\rho}}
\right].
\end{align}

Similarly, the i.i.d. expurgated exponent becomes
\begin{align}\nonumber
\Psi_{\rm ex}^{\rm iid}(\rho)
=
\sup_{s\ge 0}
-\rho \log
&\bigg[
\frac12
+
\frac12
\Big(
(1-\widetilde p)
\Big(\frac{p}{1-p}\Big)^s
\twocolbreak
\includeonetwocol{}{\quad}
+
\widetilde p
\Big(\frac{1-p}{p}\Big)^s
\Big)^{\frac1\rho}
\bigg].
\end{align}
The optimizing value is
\begin{align}
s_{\rm ex}^\star
=
\frac12
\frac{\log\frac{1-\widetilde p}{\widetilde p}}
{\log\frac{1-p}{p}} .
\end{align}
Using $s_{\rm ex}^\star$ yields
\begin{align}
\Psi_{\rm ex}(\rho)
=
-\rho \log
\left[
\frac12
\left(
1+
\left(
2\sqrt{\widetilde p(1-\widetilde p)}
\right)^{\frac1\rho}
\right)
\right].
\end{align}

For the constant-composition ensemble, the symmetry of the BSC implies that
the optimizing auxiliary function $a(\cdot)$ is constant. So the same expressions are obtained.
\color{black}

\section{Gaussian Achievable Rates}\label{Appendix_rates_Gaussian}
\subsection{GMI Rate}
 \label{Appendix_GMI_Gaussian}
Using the extension of the GMI rate in \eqref{eq_GMI} to continuous-alphabet channels, we obtain
\begin{align}
I_{\rm GMI}^{\rm G}(X;Y')
&=\frac{1}{2}\log\Big(1+\frac{(1-\alpha)P}{\sigma^2+\alpha P}\Big)
\twocolbreak
\includeonetwocol{}{\quad}
-\inf_{s\geq 0}D(u'(y'|x)||g_s(y'|x)|p_{Y'}(y')),
\label{GMI_rate_G}
\end{align}
where $p_{Y'}(y')
=\frac{1}{\sqrt{2 \pi \sigma_{y'}^2}}e^{-\frac{y'^2}{2 \sigma_{y'}^2}}$, in which $\sigma_{y'}^2=P+\sigma^2$ and
\begin{align}
\label{u_prime_G}
&u'(x|y')
=\frac{1}{\sqrt{2 \pi \sigma_{u'}^2}}e^{-\frac{
(x-\beta_{u'} y')^2}{2\sigma_{u'}^2}},\\ \label{v_prime_G}
&g_s(x|y')
=\frac{1}{\sqrt{2 \pi \sigma_{g}^2(s)}} e^{-\frac{(x-\beta_{g}(s)y')^2}{2\sigma_{g}^2(s)}},
\end{align}
where $\sigma_{u'}^2=\frac{P(\alpha P+\sigma^2)}{P+\sigma^2}$, $\beta_{u'}=\frac{P\sqrt{1-\alpha}}{P+\sigma^2}$, $\sigma_{g}^2(s)=\frac{P\sigma^2}{sP+\sigma^2}$, and $\beta_{g}(s)=\frac{sP}{sP+\sigma^2}$.
Hence, using the KL-divergence of two Gaussian distributions, we obtain
\begin{align}
&D(u'(x|y')||g_s(x|y')|p_{Y'}(y'))
\twocolbreak
\includeonetwocol{}{\quad}
=\frac{1}{2}\bigg[\frac{\sigma_{u'}^2+(\beta_{u'}-\beta_{g}(s))^2\sigma_{y'}^2}{\sigma_{g}^2(s)}+\log\bigg(\frac{\sigma_{g}^2(s)}{\sigma_{u'}^2}\bigg)-1\bigg].
\end{align}
By minimizing the KL divergence over $s$ we obtain \eqref{GMI_rate_G_R1}. Let $\Delta(s)=D\big(u'(x|y')||g_s(x|y')|p_{Y'}(y')\big)$. To obtain the optimum value of $s$ which minimizes $\Delta(s)$, we take the derivative of $\Delta(s)$ and find its roots
\begin{align}\nonumber
s_{1,2}=\tfrac{-(3-4\sqrt{1-\alpha}+\frac{2\sigma^2}{P}) \pm \sqrt{(1+\frac{2\sigma^2}{P})^2+8(1-\sqrt{1-\alpha})(1+\frac{\sigma^2}{P})}}{\frac{4P}{\sigma^2}(1-\sqrt{1-\alpha})+2},
\end{align}
where $s_1$ and $s_2$ correspond to the roots with
the plus and minus signs, respectively. Since 
\begin{align}
s_1s_2=
-\frac{\frac{2P}{\sigma^2}\sqrt{1-\alpha}}
{1+\frac{2P}{\sigma^2}(1-\sqrt{1-\alpha})}<0,
\end{align}
it follows that 
$s_1 \geq 0$ and $s_2 \leq 0$. 
Moreover, a sign analysis of $\frac{d\Delta(s)}{ds}$ shows that
$\Delta(s)$ decreases for $0\le s<s_1$ and increases for $s>s_1$. Hence, $s_1$ is the unique minimizer of $\Delta(s)$ over the feasible set $s\ge0$. 
\subsection{LM Rate}
\label{Appendix_LM_Gaussian}
Using the extension of \eqref{eq_LM} to continuous channels, we have
\begin{align}
I_{\rm LM}^{\rm G}(X;Y')
&=\frac{1}{2}\log\Big(1+\frac{(1-\alpha)P}{\sigma^2+\alpha P}\Big)
\twocolbreak
\includeonetwocol{}{\quad}
-\inf_{s\geq 0,a(.)}D\big(u'(x|y')||g'_{s,a}(x|y')|p_{Y'}(y')\big),
\label{LM_rate_G}
\end{align}
where $p_{Y'}(y')
$ and $u'(x|y')
$ are the same as in the GMI rate, and $g'_{s,a}(x|y')
$ can be simplified as
\begin{equation}
\begin{aligned}
\label{v_double_prime_G}
&g'_{s,a}(x|y')
=\frac{e^{-\big(\frac{\sigma^2+sP}{2P\sigma^2}x^2-\frac{s}{\sigma^2}xy'\big)+a(x)}}{\int_{-\infty}^{\infty} e^{-\big(\frac{\sigma^2+sP}{2P\sigma^2}\bar{x}^2-\frac{s}{\sigma^2}\bar{x}y'\big)+a(\bar{x})}d \bar{x}},
\end{aligned}
\end{equation}
To maintain the Gaussian form for $g'_{s,a}(x|y')$, $a(x)$ should be a quadratic function of $x$, i.e., $a(x)=a_1 x^2 + a_2 x$.
We note that adding a fixed value to $a(x)$ will not alter $g'_{s,a}(x|y')$, and thus we avoid such additions. Using this function, we derive a lower bound on the LM rate. For this specific choice of $a(x)$, $g'_{s,a}(x|y')$ can be simplified as
\begin{align}
\label{u_prime_G2}
g'_{s,a}(x|y')
& =\frac{1}{\sqrt{2 \pi \sigma_{g'}^2(s)}}e^{-\frac{\left(x-\beta_{g',1}(s)y'-\beta_{g',2}(s)\right)^2}{2\sigma_{g'}^2(s)}},
\end{align}
where $\sigma_{g'}^2(s)=\frac{P\sigma^2}{sP+\sigma^2-2P\sigma^2 a_1}$, $\beta_{g',1}(s)=\frac{sP}{sP+\sigma^2-2P\sigma^2 a_1}$, and $\beta_{g',2}(s)=\frac{a_2 P \sigma^2}{sP+\sigma^2-2P\sigma^2 a_1}$. Hence, 
\begin{align}
&D(u'(x|y')||g'_{s,a}(x|y')|p_{Y'}(y'))
\twocolbreak
\includeonetwocol{}{\quad}
=\frac{1}{2}\bigg[\frac{\sigma_{u'}^2+(\beta_{u'}-\beta_{g',1})^2(s)\sigma_{y'}^2+\beta_{g',2}^2(s)}{\sigma_{g'}^2(s)}
\twocolbreak
\includeonetwocol{}{\qquad}
+\log\bigg(\frac{\sigma_{g'}^2(s)}{\sigma_{u'}^2}\bigg)-1\bigg].
\end{align}
According to the above equation, it can be easily seen that the above KL divergence is minimized for ${a_2=0}$. Hence, 
${\beta_{g',2}(s)=0}$. 
Further, by choosing ${s=\frac{\sigma^2 \sqrt{1-\alpha}}{\alpha P+\sigma^2}}$ and ${a_1=\frac{1-\sqrt{1-\alpha}}{2(1-\alpha)(\alpha P +\sigma^2)}}$, the above KL divergence becomes zero and we get the matched rate.

\bibliographystyle{IEEEtran}
\bibliography{Reftest}

\end{document}